\documentclass[11pt]{article}

\usepackage[margin=1in]{geometry}

\usepackage{amsfonts,amssymb,amsmath,amsthm,mathtools}
\usepackage[usenames,dvipsnames]{xcolor}
\usepackage{booktabs}
\usepackage{enumitem}
\usepackage{tocloft}
\usepackage{thm-restate}
\usepackage{dsfont}
\usepackage{tikz}

\usepackage[pagebackref]{hyperref}
\hypersetup{
    pdftitle={Separations in communication complexity using cheat sheets and information complexity}, 
    pdfauthor={}, 
    colorlinks=true, 
    linkcolor=blue, 
    citecolor=blue, 
    filecolor=blue, 
    urlcolor=blue 
}

\allowdisplaybreaks

\renewcommand{\backref}[1]{}

\renewcommand{\backrefalt}[4]{%
\small
\ifcase #1 %
\or
[p.\ #2]%
\else
[pp.\ #2]%
\fi}

\makeatletter
\renewcommand{\paragraph}{%
  \@startsection{paragraph}{4}%
  {\z@}{3ex \@plus .5ex \@minus .3ex}{-1em}%
  {\normalfont\normalsize\bfseries}%
}
\makeatother

\makeatletter
\newcommand{\para}{%
  \@startsection{paragraph}{4}%
  {\z@}{1.5ex \@plus .5ex \@minus .3ex}{-1em}%
  {\normalfont\normalsize\bfseries}%
}
\makeatother

\newtheorem{theorem}{Theorem}
\newtheorem{lemma}{Lemma}
\newtheorem{corollary}[theorem]{Corollary}

\newtheorem{fact}[lemma]{Fact}
\newtheorem{claim}[lemma]{Claim}

\theoremstyle{definition}
\newtheorem{definition}[lemma]{Definition}
\newtheorem{remark}[lemma]{Remark}

\newtheoremstyle{part}
  {-0.7\topsep}   
  {\topsep}   
  {\itshape}  
  {0pt}       
  {\bfseries} 
  {.}         
  {5pt plus 1pt minus 1pt} 
  {}          

\theoremstyle{part}
\newtheorem{factpart}{Fact}[lemma]

\newcommand{\eq}[1]{\hyperref[eq:#1]{(\ref*{eq:#1})}}
\renewcommand{\sec}[1]{\hyperref[sec:#1]{Section~\ref*{sec:#1}}}
\newcommand{\thm}[1]{\hyperref[thm:#1]{Theorem~\ref*{thm:#1}}}
\newcommand{\lem}[1]{\hyperref[lem:#1]{Lemma~\ref*{lem:#1}}}
\newcommand{\prop}[1]{\hyperref[prop:#1]{Proposition~\ref*{prop:#1}}}
\newcommand{\cor}[1]{\hyperref[cor:#1]{Corollary~\ref*{cor:#1}}}
\newcommand{\fig}[1]{\hyperref[fig:#1]{Figure~\ref*{fig:#1}}}
\newcommand{\tab}[1]{\hyperref[tab:#1]{Table~\ref*{tab:#1}}}
\newcommand{\alg}[1]{\hyperref[alg:#1]{Algorithm~\ref*{alg:#1}}}
\newcommand{\app}[1]{\hyperref[app:#1]{Appendix~\ref*{app:#1}}}
\newcommand{\defn}[1]{\hyperref[def:#1]{Definition~\ref*{def:#1}}}
\newcommand{\clm}[1]{\hyperref[clm:#1]{Claim~\ref*{clm:#1}}}
\newcommand{\fct}[1]{\hyperref[fact:#1]{Fact~\ref*{fact:#1}}}

\newcommand*{\fullref}[1]{\hyperref[{#1}]{\autoref*{#1}:~\nameref*{#1}}}

\newcommand*{\fullbref}[1]{\hyperref[{#1}]{\autoref*{#1} (\nameref*{#1})}}

\newcommand{\B}{\{0,1\}}

\newcommand{\AND}{\textsc{And}}

\newcommand{\OR}{\textsc{Or}}
\newcommand{\pOR}{\textsc{Pr-Or}}
\newcommand{\FORR}{\textsc{Forrelation}}
\newcommand{\SIMON}{\textsc{Simon}}

\newcommand{\IP}{\textsc{IP}}
\newcommand{\STR}{\textsc{Str}}
\newcommand{\Addr}{\textsc{Addr}}
\newcommand{\XOR}{\textsc{Xor}}
\newcommand{\TR}{\textsc{Tr}}
\newcommand{\pTR}{\textsc{Ptr}}
\newcommand{\Disj}{\textsc{Disj}}

\newcommand{\X}{\mathcal{X}}
\newcommand{\Y}{\mathcal{Y}}
\newcommand{\Z}{\mathcal{Z}}

\newcommand{\CS}{\mathrm{CS}}

\newcommand{\tO}{\widetilde{O}}
\newcommand{\tOmega}{\widetilde{\Omega}}

\newcommand{\dt}{\mathrm{dt}}
\renewcommand{\(}{\left(}
\renewcommand{\)}{\right)}

\renewcommand{\deg}{\operatorname{deg}^+}

\newcommand{\bR}{\mathbb{R}}

\DeclareMathOperator{\dom}{{dom}}

\DeclareMathOperator{\D}{D}
\DeclareMathOperator{\R}{R}
\DeclareMathOperator{\Q}{Q}
\DeclareMathOperator{\N}{N}
\DeclareMathOperator{\C}{C}
\DeclareMathOperator{\IC}{IC}
\DeclareMathOperator{\UN}{UN}

\newcommand{\Ddt}{\D^\dt}
\newcommand{\Rdt}{\R^\dt}
\newcommand{\Qdt}{\Q^\dt}
\newcommand{\QEdt}{\Q_E^\dt}
\newcommand{\UNdt}{\UN^\dt}

\newcommand{\G}{\mathcal{G}}
\newcommand{\T}{\mathcal{T}}
\newcommand{\x}{\mathbf{x}}
\newcommand{\y}{\mathbf{y}}
\newcommand{\fu}{\mathbf{u}}
\newcommand{\fv}{\mathbf{v}}

\newcommand{\defeq}{\coloneqq} 

\newcommand{\eps}{\varepsilon}
\renewcommand{\epsilon}{\varepsilon}
\renewcommand{\Pr}{\mathrm{Pr}}
\newcommand{\E}{\mathbb{E}}

\newcommand{\I}{\mathbb{I}}
\newcommand{\DIV}{\mathbb{D}}
\newcommand{\h}{\mathrm{h}}
\newcommand{\mH}{\mathbb{H}}
\newcommand{\VR}{\Delta}

\newcommand{\err}{\mathrm{err}}

\newcommand{\CC}{\mathrm{CC}}

\begin{document}

\title{\bfseries Separations in communication complexity\\ using cheat sheets and information complexity}

\author{
Anurag Anshu\footnote{Centre for Quantum Technologies, National University of Singapore, Singapore. \texttt{a0109169@u.nus.edu}} \and 
Aleksandrs Belovs\footnote{CWI, Amsterdam, The Netherlands. {\tt stiboh@gmail.com}} \and 
Shalev Ben-David\footnote{Massachusetts Institute of Technology. \texttt{shalev@mit.edu}} \and 
Mika G\"o{\"o}s\footnote{University of Toronto. \texttt{mgoos@cs.toronto.edu}} \and
Rahul Jain\footnote{Centre for Quantum Technologies, National University of Singapore and MajuLab, UMI 3654, 
Singapore. \texttt{rahul@comp.nus.edu.sg}} \and 
Robin Kothari\footnote{Massachusetts Institute of Technology. \texttt{rkothari@mit.edu}} \and 
Troy Lee\footnote{SPMS, Nanyang Technological University and Centre for Quantum Technologies and MajuLab, UMI 3654, Singapore. {\tt troyjlee@gmail.com}} \and 
Miklos Santha\footnote{IRIF, Universit\'e Paris Diderot, CNRS, 75205 Paris, France;  and
Centre for Quantum Technologies, National University of Singapore,
Singapore. {\tt miklos.santha@liafa.univ-paris-diderot.fr}}
}

\hypersetup{pageanchor=false} 
\date{}
\maketitle

\begin{abstract}
While exponential separations are known between quantum and randomized communication complexity for partial 
functions (Raz, {\footnotesize STOC 1999}), 
the best known separation between these measures for a total function is quadratic, 
witnessed by the disjointness function.  We give the first super-quadratic separation between quantum and randomized 
communication complexity for a total function, giving an example exhibiting a power $2.5$ gap.  We further 
present a $1.5$ power separation between exact quantum and randomized communication complexity, improving 
on the previous $\approx 1.15$ separation by Ambainis ({\footnotesize STOC 2013}).  Finally, we present a nearly 
optimal quadratic separation between randomized communication complexity and the logarithm of the partition number, 
improving upon the previous best power $1.5$ separation due to G\"o\"os, Jayram, Pitassi, and Watson. 

Our results are the communication analogues of separations in query complexity proved using the recent cheat sheet 
framework of Aaronson, Ben-David, and Kothari ({\footnotesize STOC 2016}).
Our main technical results are randomized communication and information complexity lower bounds for a family of 
functions, called lookup functions, that generalize and port the cheat sheet framework to communication complexity.
\end{abstract}

\thispagestyle{empty}
\clearpage
\section{Introduction}
\label{sec:intro}
\hypersetup{pageanchor=true} 
\setcounter{page}{1}

Understanding the power of different computational resources is one of the primary aims of complexity theory.  
Communication complexity provides an ideal setting to study these questions, as it is a nontrivial model for which we are 
still able to show interesting lower bounds.  Moreover, lower bounds in communication complexity have applications to 
many other areas of complexity theory, for example yielding lower bounds for circuits, data structures, streaming algorithms, 
property testing, and linear and semi-definite programs. 

In communication complexity, two players Alice and Bob are given inputs $x\in\X$ and $y\in\Y$ respectively, and their task 
is to compute a known function $F\colon\X \times \Y \to \{0,1,*\}$ while minimizing the number of bits communicated between 
them.  We call such a function a communication function. The players only need to be correct on inputs $(x,y)$ for which $F(x,y) \in \B$.  The function is called total if 
$F(x,y) \in \B$ for all $(x,y) \in \X \times \Y$, and otherwise is called partial.  

A major question in communication complexity is what advantage players who exchange quantum messages can 
achieve over their classical counterparts.  We will use $\R(F)$ and $\Q(F)$ to denote bounded-error (say $1/3$)
public-coin randomized  and bounded-error 
quantum communication complexities of~$F$, respectively.  We also use $\D(F)$ for the deterministic communication
complexity and $\Q_E(F)$ for the exact quantum communication complexities of~$F$, respectively.  Note the easy 
relationships $\D(F) \ge \R(F) \ge \Q(F)$ and $\D(F) \ge \Q_E(F) \ge \Q(F)$.

There are examples of \emph{partial} functions $F$ for which $\Q(F)$ is exponentially smaller than $\R(F)$ \cite{Raz99}.  
For total functions, however, it is an open question if $\Q(F)$ and $\R(F)$ are always polynomially related.  On the other 
hand, the largest separation between these measures is quadratic, witnessed by the disjointness function which satisfies
$\R(\Disj_n) = \Omega(n)$ \cite{KS92,Raz92} and $\Q(\Disj_n) = O(\sqrt{n})$ \cite{BCW98,AA03}.  Our first result gives the first super-quadratic separation between $\Q(F)$ and $\R(F)$ for a total function.
\begin{restatable}{theorem}{qvsr} \label{thm:q-vs-r}
There exists a total function $F\colon\X \times \Y \to \B$ with $\R(F) = \tOmega(\Q(F)^{2.5})$.
\end{restatable}

In fact, we establish a power $2.5$ separation between $\Q(F)$ and information complexity~\cite{BJKS04}, a lower bound technique for randomized communication complexity (defined in \sec{prelim}).

We also give a $1.5$ power separation between randomized communication complexity and \emph{exact} quantum 
communication complexity.  This improves the previous best separation of $\approx 1.15$ due to Ambainis \cite{Amb13}.
\begin{restatable}{theorem}{qevsr} \label{thm:qe-vs-r}
There exists a total function $F\colon\X \times \Y \to \B$ with $\R(F) = \tOmega(\Q_E(F)^{1.5})$.
\end{restatable}

Another interesting question in communication complexity is the power of different lower bound techniques.  
After years of work on randomized communication complexity lower bounds, there are essentially two lower bound 
techniques that stand at the top of the heap, the aforementioned information complexity~\cite{BJKS04} and the partition bound \cite{JK10}.  Both of these 
techniques are known to dominate many other techniques in the literature, such as the smooth rectangle bound, 
corruption bound, discrepancy, etc., but the relationship between them is not yet known. For deterministic protocols, a bound even more powerful 
than the partition bound, is the logarithm of the partition number.  
The partition number, denoted $\chi(F)$, is the smallest number of $F$-monochromatic 
rectangles in a partition of $\X \times \Y$ (see \sec{prelim} for more precise definitions).  We use the notation 
$\UN(F) = \log \chi(F)$, where $\UN$ stands for unambiguous nondeterministic communication complexity.  

Showing separations between $\R(F)$ and $\UN(F)$ is very difficult because there are few techniques available to 
lower bound $\R(F)$ that do not also lower bound $\UN(F)$.  Indeed, until recently only a factor 2 separation was known 
even between $\D(F)$ and $\UN(F)$, shown by Kushilevitz, Linial, and Ostrovsky \cite{KLO99}.  This changed with the breakthrough work of G\"o\"os, Pitassi, and Watson \cite{GPW15}, who 
exhibited a total function $F$ with $\D(F) = \tOmega(\UN(F)^{1.5})$.  
Ambainis, Kokainis and Kothari \cite{AKK16} improved this by constructing a total function 
$F$ with $\D(F) \ge \UN(F)^{2-o(1)}$.  This separation is nearly optimal as 
Aho, Ullman, and Yannakakis \cite{AUY83} showed 
$\D(F) =O(\UN(F)^2)$ for all total $F$.

G\"o\"os, Jayram, Pitassi, and Watson~\cite{GJPW15} improved the original \cite{GPW15} separation in 
a different direction, constructing a total $F$ for which $\R(F) = \tOmega(\UN(F)^{1.5})$.  
In this paper, we achieve a nearly optimal separation between these measures.
\begin{restatable}{theorem}{unvsr} \label{thm:un-vs-r}
There exists a total function $F\colon\X \times \Y \to \B$ with $\R(F) \ge \UN(F)^{2-o(1)}$.
\end{restatable}
In particular, this means the partition bound can be quadratically smaller than $\R(F)$, since the partition bound is at most 
$\UN(F)$.

\subsection{Comparison with prior work}
The model of query complexity provides insight into communication complexity and is usually easier to understand.  
Many theorems in query complexity have analogous results in communication complexity.  There is also a more precise connection between these models, which we now explain.  For a function 
$f\colon \B^n \rightarrow \B$, let $\Ddt(f)$ be the deterministic query complexity of $f$, the minimum number of queries 
an algorithm needs to the bits of the input $x$ to compute $f(x)$, in the worst case.  Similarly, let $\Rdt(f)$, $\Qdt(f)$, and 
$\UNdt(f)$ denote 
the randomized, quantum and unambiguous nondeterministic query complexities of $f$.

Any function $f$ can be turned into a communication problem by composing it with a communication ``gadget'' 
$G\colon \X \times \Y \rightarrow \B$.  On input $((x_1, \ldots, x_n), (y_1, \ldots, y_n))$ the function $f \circ G$ evaluates to 
$f(G(x_1, y_1), \ldots, G(x_n,y_n))$.  It is straightforward to see that $\D(f \circ G) \le \Ddt(f) \D(G)$, and analogous 
results hold for $\UN(f\circ G)$, $\R(f \circ G)$, and $\Q(f \circ G)$ (with extra logarithmic factors).  

The reverse direction, that is, lower bounding the communication complexity of $f \circ G$ in terms of the query 
complexity of $f$ is not always true, but can hold for specific functions $G$.  Such results are called ``lifting'' theorems 
and are highly nontrivial.  
G\"o\"os, Pitassi, and Watson~\cite{GPW15}, building on work of  Raz and McKenzie~\cite{RazM99}, show a general lifting theorem for deterministic query complexity: for a specific 
$G\colon\B^{20\log n} \times \B^{n^{20}} \to \B$, with $\D(G) = O(\log n)$, it holds that $\D(f\circ G) = \Omega(\Ddt(f)\log n)$, 
for any $f\colon\B^n \to \B$.  

This allowed them to achieve their separation between $\D$ and $\UN$ by first showing the analogous result in the 
query world, i.e., exhibiting a function $f\colon\B^n\to\B$ with $\Ddt(f) = \tOmega(\UNdt(f)^{1.5})$, and then using the
lifting theorem to achieve the same separation for a communication problem.
The work of Ambainis, Kokainis, and Kothari~\cite{AKK16} followed the same plan and obtained their 
communication complexity separation by improving the query complexity separation of \cite{GPW15} to 
$\Ddt(f) \ge \UNdt(f)^{2-o(1)}$.

For separations against randomized communication complexity, as in our case, the situation is different.  Analogs 
of our results have been shown in query complexity.   Aaronson, Ben-David, and Kothari \cite{ABK16} defined a 
transformation of a Boolean function, which they called the ``cheat sheet technique.''   This transformation takes a 
function $f$ and returns a cheat sheet function, $f_{\CS}$, whose randomized query complexity is at least that of $f$.  
They used this method to give a total function $f$ with $\Rdt(f) = \tOmega(\Qdt(f)^{2.5})$.
The cheat sheet technique is also used in \cite{AKK16} to show the query analog of
our~\autoref{thm:un-vs-r}, giving an $f$ with $\Rdt(f) \ge \UNdt(f)^{2-o(1)}$.  These results, however, do not immediately imply similar results for communication complexity as no general theorem is known
to lift randomized query lower bounds to randomized communication 
lower bounds.  Such a theorem could hold and is an interesting open problem.

The most similar result to ours is that of G\"o\"os, Jayram, Pitassi, and Watson~\cite{GJPW15} who show 
$\R(F) = \tOmega(\UN(F)^{1.5})$. While the query analogue $\Rdt(f) = \tOmega(\UNdt(f)^{1.5})$ was not hard to show, the communication separation required developing new communication complexity techniques.  We similarly work directly 
in the setting of communication complexity, as described next. 

\subsection{Techniques}
While a lifting theorem is not known for randomized query complexity, a lifting theorem is known for a stronger model
known as \emph{approximate conical junta degree}, denoted $\deg_{1/3}(f)$ (formally defined 
in~\autoref{sec:juntadegree}).  This is a 
query measure that satisfies $\deg_{1/3}(f) \le \R(f)$ and has a known lifting theorem~\cite{GLM+15} 
(see~\autoref{thm:simu}).  
The first idea to obtain our theorems would be to show (say) that $\deg_{1/10}(\neg f_{\CS}) = \tOmega(\deg_{1/3}(f))$\footnote{We negate the function $f_\CS$ because the obvious statement $\deg_{1/10}({f_{\CS}}) = \tOmega(\deg_{1/3}(f))$ is false in general.} and to 
use this lifting theorem.  We were not able to show such a theorem, however, in part because $\deg_\epsilon(f)$ does not 
behave well with respect to the error parameter $\epsilon$.

Instead we work directly in the setting of communication complexity.  
We show randomized communication lower bounds for a broad family of communication functions called lookup functions.  
For intuition about a lookup function, consider first the query setting and 
the familiar address function $\Addr \colon \B^{c + 2^c} \rightarrow \B$.  Think of the input as divided into two parts, 
$\x=(x_1, \ldots, x_c) \in \B^c$ and the data $\fu=(u_0, \ldots, u_{2^c-1}) \in \B^{2^c}$.  The bit string $\x$ is interpreted as an integer 
$\ell \in \{0, \ldots, 2^c-1\}$ and the output of $\Addr(\x,\fu)$ is $u_\ell$.

A natural generalization of this problem is to instead have a function\footnote{For simplicity we restrict to total functions 
here. The full definition (\autoref{def:lookup}) also allows for partial functions.} $f\colon \B^n \rightarrow \B$ and functions 
$g_j \colon \B^{cn} \times \B^m \rightarrow \B$ for $j \in \{0,\ldots, 2^c-1\}$.  Now the input consists of $\x=(x_1, \ldots, x_c)$ 
where each $x_i \in \B^n$, and $\fu=(u_0, \ldots, u_{2^c-1})$ where each $u_j \in \B^m$.  An address 
$\ell \in  \{0,\ldots, 2^c-1\}$ 
is defined by the string $(f(x_1), \ldots, f(x_c))$, and the output of the function is $g_\ell(\x, u_\ell)$.  Call such 
a function a $(f,\{g_0, \ldots, g_{2^c-1}\})$-lookup function.  The cheat sheet framework of~\cite{ABK16} naturally fits 
into this framework: the cheat sheet function $f_{\CS}$ of $f$ is a lookup function where 
$g_\ell(x_1, \ldots, x_c, u_\ell)=1$ if and only if $u_\ell$ provides certificates that $f(x_i)=\ell_i$ for each $i \in [c]$.  

This idea also extends to communication complexity where one can define a $(F,\G)$-lookup function in the same way,  
with $F$ a communication function and $\G=\{G_0, \ldots, G_{2^c-1}\}$ a family of communication functions.  Our main 
technical 
theorem (\autoref{thm:ICCS}) states that, 
under mild conditions on the family $\G$, the randomized communication complexity of the $(F,\G)$-lookup function is at 
least that of $F$.  
To prove the separation of~\autoref{thm:q-vs-r}, we take the function $f=\SIMON_n\circ \OR_n\circ \AND_n$ and let $F$ be 
$f$ composed with the inner product communication gadget.  We define the family of functions $\G$ in a similar 
fashion as in 
the cheat sheet framework.  We show a randomized communication lower bound on $F$ using the approximate conical 
junta degree and the lifting theorem of~\cite{GLM+15}.  The separation of~\autoref{thm:qe-vs-r} follows a 
very similar plan, starting instead with the query function $h = \pOR_n \circ \AND_m$ for $m = \Theta(\sqrt{n})$, 
where $\pOR_n$ is a promise version of the $\OR_n$ function restricted to inputs of Hamming weight $0$ or $1$.  

Moving on to our third result (\thm{un-vs-r}), we find that just having a lower bound on the randomized 
communication complexity of a $(F, \G)$-lookup function is not enough to obtain the separation. The query analogue of \thm{un-vs-r} \cite{AKK16} relies on repeatedly composing a function with $\AND_n$ (or $\OR_n$), which raises its randomized query complexity by $\Omega(n)$. More precisely, it relies on the fact that $\Rdt(\AND_n\circ f) = \Omega(n \Rdt(f))$. However, the analogous communication complexity claim, $\R(\AND_n \circ F) = \Omega(n\R(F))$, is false. For a silly example, if $F$ itself is $\AND_n$ (under some bipartition of input bits), then $\R(\AND_n\circ F)\leq\D(\AND_{n^2})=O(1)$. Another example is if $F\colon\B\times \B \to \B$ is the equality function on $1$ bit, then $\R(\AND_n \circ F)=O(1)$, since this is the equality function on $n$ bits.

To circumvent this issue, we use information complexity instead of randomized communication complexity. Let $\IC(F)$ denote the information complexity of a function $F$ (defined in \sec{prelim}). Information complexity, or more precisely one-sided information complexity, satisfies a composition theorem for the $\AND_n$ function (\fct{AND}). While one-sided information complexity upper bounds can be converted to information complexity upper bounds (\fct{IC_onesided}), the conversion also requires upper bounding the communication complexity of the protocol. This makes the argument delicate and requires simultaneously keeping track of the information complexity and communication complexity throughout the argument. Informally, we show the following theorem.

\begin{theorem}[informal]
For any function $F$, and any family of functions $\G=\{G_0, \ldots, G_{2^c-1}\}$ let $F_\G$ be the $(F,\G)$-lookup function.  Provided $\G$ satisfies certain mild technical conditions, $\R(F_\G) = \tOmega(\R(F))$ and $\IC(F_\G) = \tOmega(\IC(F))$.
\end{theorem}

We prove this formally as \thm{ICCS} in \sec{cheat}. This is the most technical part of the paper, requiring all the preliminary facts and notation set up in \autoref{sec:info} and \autoref{sec:comm}. The proof relies on an information theoretic argument that establishes that a correct protocol for $F_\G$ already has enough information to compute one copy of the base function $F$.

\section{Preliminaries and notation}
\label{sec:prelim}

In this paper we denote query complexity (or decision tree complexity) measures using the superscript $\dt$. For example, the deterministic, bounded-error randomized, exact quantum, and bounded-error quantum query complexities of a function $f$ are denoted $\Ddt(f)$, $\Rdt(f)$, $\QEdt(f)$, and $\Qdt(f)$ respectively. We refer the reader to the survey by Buhrman and de Wolf \cite{BdW02} for formal definitions of these measures. 

A function $f\colon\B^n \to \{0,1,*\}$ is said to be a total function if $f(x)\in \B$ for all $x\in\B^n$ and is said to be partial otherwise. We define $\dom(f) \defeq \{x:f(x)\neq *\}$ to be the set of valid inputs to $f$. An algorithm computing $f$ is allowed to output an arbitrary value for inputs outside 
$\dom(f)$.
$\AND_n$ and $\OR_n$ denote the $\AND$ and $\OR$ functions on $n$ bits, defined as $\AND_n(x_1,\ldots, x_n) \defeq \bigwedge_{i=1}^n x_i$ and $\OR_n(x_1,\ldots, x_n) \defeq \bigvee_{i=1}^n x_i$. In general, $f_n$ denotes an $n$-bit function.

In communication complexity, we wish to compute a function $F\colon\X \times \Y \to \{0,1,*\}$ for some finite sets $\X$ and $\Y$, where the inputs $x\in \X$ and $y\in\Y$ are given to two players Alice and Bob, while minimizing the communication between the two. As in query complexity, $F$ is total if $F(x,y) \in \B$ for all $(x,y)\in \X \times \Y$ and is partial otherwise. We define
$\dom(F) \defeq \{(x,y):F(x,y)\neq *\}$. As before a correct protocol may behave arbitrarily on inputs outside $\dom(F)$. Formal definitions of the measures studied here can be found in the textbook by Kushilevitz and Nisan \cite{KN06}.

For a function $f\colon\B^n \rightarrow \{0,1,*\}$ we let $f^c$ denote the function $f^c\colon \B^{nc} \rightarrow \{0,1,*\}^c$ 
where $f^c(x_1, \ldots, x_c) = (f(x_1), \ldots, f(x_c))$.  Note that $\dom(f^c) = \dom(f)^c$.  For 
a communication function $F\colon\X \times \Y \to \B$ we let $F^c \colon \X^c \times \Y^c \rightarrow \B^c$ be 
$F^c((x_1, \ldots, x_c), (y_1, \ldots, y_c)) = (F(x_1, y_1), \ldots, F(x_c,y_c))$.

We use $\D(F)$ to denote the deterministic communication complexity of $F$, the minimum number of bits exchanged in a deterministic communication protocol that correctly computes $F(x,y)$ for all inputs in $\dom(F)$. Public-coin randomized and quantum (without entanglement) communication complexities, denoted $\R(F)$ and $\Q(F)$, are defined similarly except the protocol may now err with probability at most $1/3$ on any input and may use random coins or quantum messages respectively. Exact quantum communication complexity, denoted $\Q_E(F)$, is defined similarly, except it must output the correct answer with certainty.

We use $\N(F)$ and $\UN(F)$ to denote the nondeterministic (or certificate) complexity of $F$ and the unambiguous nondeterministic complexity of $F$ respectively. $\UN(F)$ equals $\log \chi(F)$, where $\chi(F)$ is the partition number of $F$, the least number of monochromatic rectangles in a partition (or disjoint cover) of $\X \times \Y$. We now define these measures formally.

Given a partial function $F\colon\X \times \Y \to \{0,1,*\}$ and $b\in\B$, a $b$-monochromatic rectangle is a set $A\times B$ with $A\subseteq\X$ and $B\subseteq\Y$ such that all inputs in $A \times B$ evaluate to $b$ or $*$ on $F$. A $b$-cover of $F$ is a set of $b$-monochromatic rectangles that cover all the $b$-inputs (i.e., inputs that evaluate to $b$ on $F$) of $F$. 
If the rectangles form a partition of the $b$-inputs, we say that the cover is unambiguous. 
Given a $b$-cover of $F$, a $b$-certificate for input $(x,y)$ is the label of a rectangle containing $(x,y)$ in the $b$-cover.
The $b$-cover number $\C_b(F)$ is the size of the smallest $b$-cover, and we set $\N_b(F) \coloneqq \lceil \log \C_b(F)\rceil$. The 
nondeterministic complexity of $F$ is $\N(F) \coloneqq \max \{\N_0(F), \N_1(F)\}$. The quantities $\UN_b(F)$ and the
unambiguous non-deterministic complexity $\UN(F)$ are defined analogously from partitions. 

It is useful to interpret a $b$-certificate for $(x,y)\in\dom(F)$ as a message that an all-powerful prover can send to the players to convince each of them that $F(x,y) = b$. In this interpretation,  $\N_b(F)$ is the minimum over prover strategies of the maximum length of a message taken over all inputs.
Similarly, $\UN_b(F)$ is the maximum length of a message when, in addition, for every input in $\dom(F)$,
there is exactly one certificate the prover can send.  

We also use $\IC(F)$ to denote the information complexity of $F$, defined formally in~\autoref{sec:comm}. Informally, the information complexity of a function $F$ is the minimum amount of information about their inputs that the players have to reveal to each other to compute $F$. $\IC(F)$ is a lower bound on randomized communication complexity, because the number of bits communicated in a protocol is certainly an upper bound on the information gained by any player, since $1$ bit of communication can at most have $1$ bit of information.

In \sec{info} and \sec{comm} we cover some preliminaries needed to prove \thm{ICCS}.

\subsection{Information theory}
\label{sec:info}
We now introduce some definitions and facts from information theory. Please refer to the textbook by Cover and Thomas~\cite{CT06} for an excellent introduction to information theory. 

For a finite set $S$, we say $P\colon S \to \bR^+$ is a probability distribution over $S$ if $\sum_{s\in S} P(s) = 1$. 
For correlated random variables $XYZ \in \X \times \Y \times \Z$, we use the same symbol represent the random variable and its distribution. If $\mu$ is a distribution over $\X$, we say $X \sim \mu$ to represent that $X$ is distributed according to $\mu$ and $X \sim Y$ to represent that $X$ and $Y$ are similarly distributed. We use $Y^x$ as shorthand for $(Y~|~(X=x))$.   We define the joint random variable $X \otimes Y \in \X \times \Y$ as
$$ \Pr(X \otimes Y = (x,y)) = \Pr(X = x) \cdot \Pr(Y=y) .$$
We call $X$ and $Y$ independent random variables if  $XY \sim X \otimes Y$. 

A basic fact about random variables is Markov's inequality.
We'll often make use of one particular corollary of the inequality, which we state
here for convenience.

\begin{fact}[Markov's Inequality]\label{fact:markov_ineq}
	If $Z$ is a random variable over $\bR^+$, then for any $c\geq 1$,
	\[\Pr(Z\geq c\E[Z])\leq \frac{1}{c}.\]
	In particular, if $f$ is a function mapping the domains of $X$ and $Y$ to
	$\bR^+$, then
	\[\Pr_{x\leftarrow X}(\E_{y\leftarrow Y^x}[f(x,y)]>\alpha)<\beta\quad\Rightarrow\quad
	\Pr_{(x,y)\leftarrow XY}(f(x,y)>100\alpha)<\beta+0.01.\]
\end{fact}

\begin{proof}
	To see that the first equation holds, note that if the elements $z\in\Z$
	that are larger than $c\E[Z]$ have probability mass more than $1/c$, then they
	contribute more than $\E[Z]$ to the expectation of $Z$; but since the domain
	of $Z$ is non-negative, this implies the expectation of $Z$ is larger than $\E[Z]$,
	which is a contradiction.
	
	Now suppose that $\Pr_{x\leftarrow X}(\E_{y\leftarrow Y^x}[f(x,y)]>\alpha)<\beta$.
	We classify the elements $x\in \X$ into two types: the ``bad'' ones, which satisfy
	$\E_{y\leftarrow Y^x}[f(x,y)]>\alpha$, and the ``good'' ones, which satisfy
	$\E_{y\leftarrow Y^x}[f(x,y)]\leq\alpha$. Note that the probability that an $x$ sampled
	from $X$ is bad is less than $\beta$. For good $x$, we have
	$\Pr_{y\leftarrow Y^x}(f(x,y)>100\alpha)\leq 0.01$ by Markov's inequality above
	(using the fact that $f(x,y)$ is non-negative and $\E_{y\leftarrow Y^x}[f(x,y)]\leq\alpha$). Since the probability of a bad $x$ is less than $\beta$
	and for good $x$ the equation $f(x,y)\leq 100\alpha$
	only fails with probability $0.01$ (over choice of $y\leftarrow Y^x$), we conclude
	\[\Pr_{(x,y)\leftarrow XY}(f(x,y)>100\alpha)<\beta+0.01\]
	as desired.
\end{proof}

\subsubsection{Distance measures}

We now define the main distance measures we use and some properties of these measures.

\begin{definition}[Distance measures]
Let $P$ and $Q$ be probability distributions over $S$. We define the following distance measures between distributions.
\begin{align*}
\text{Total variation distance:}& \quad \VR(P,Q) \defeq \max_{T \subseteq S}  \sum_{s \in T}\big(P(s) - Q(s)\big) = \frac{1}{2}\sum_{s\in S} |P(s) - Q(s)|.  \\
\text{Hellinger distance:}& \quad \h(P,Q) \defeq 
\frac{1}{\sqrt{2}} \sqrt{\sum_{s \in S}\left(\sqrt{P(s)}-\sqrt{Q(s)}\right)^2 } . 
\end{align*}
\end{definition}

Note that this definition can be extended to arbitrary functions $P\colon S \to \bR^+$ and  $Q\colon S \to \bR^+$. However, when $P$ and $Q$ are probability distributions these measures are between $0$ and $1$. These extremes are achieved when $P=Q$ and when $P$ and $Q$ have disjoint support, respectively. Conveniently, these measures are closely related and are interchangeable up to a quadratic factor. 

\begin{fact}[Relation between $\VR$ and $\h$]\label{fact:relation-inf}
Let $P$ and $Q$ be probability distributions. Then
\begin{align*}
&\quad \frac{1}{\sqrt{2}}\VR(P,Q) \leq \h(P,Q) \leq \sqrt{ \VR(P,Q)}. 
\end{align*}
\end{fact}
\begin{proof}
This follows from~\cite[Theorem $15.2$, p.~$515$]{ADasG11}. In this reference, the quantity $\sqrt{2}\cdot\h(P,Q)$ is used for Hellinger distance.
\end{proof}

In this paper, we only use Hellinger distance when we invoke \fullbref{fact:pyth}, a key step in the proof of \thm{ICCS}. Hence we do not require any further properties of this measure.

On the other hand, total variation distance satisfies several useful properties that we use in our arguments. We review some of its basic properties below.

\begin{fact}[Facts about $\VR$] Let $P$, $P'$, $Q$, $Q'$, and $R$ be probability distributions and let $XY \in \X \times \Y$ and $X'Y'$ in $\X \times \Y$ be correlated random variables. Then we have the following facts.
\end{fact}
\begin{factpart}[Triangle inequality]\label{fact:triangle}
$\VR(P,Q) \leq \VR(P,R) + \VR(R,Q).$
\end{factpart}
\begin{factpart}[Product distributions]\label{fact:prod} 
$\VR(P\otimes Q, P' \otimes Q')   \leq  \VR(P, P') + \VR(Q, Q').$
Additionally, if $Q= Q'$ then $\VR(P\otimes Q, P' \otimes Q')  = \VR(P, P')$. 
\end{factpart}
\begin{factpart}[Monotonicity]\label{fact:monoTV} 
$\VR(XY, X'Y')   \geq  \VR(X, X').$
\end{factpart}
\begin{factpart}[Partial measurement]\label{fact:subadd} 
If $X \sim X'$, then $\VR(XY, X'Y') = \E_{x \leftarrow X} [\VR(Y^x, Y'^x)].$
\end{factpart}

\begin{proof}
These facts are proved as follows.
\renewcommand{\theenumi}{\bfseries \Alph{enumi}}
\begin{enumerate}
\item Let $P$, $Q$, and $R$ be distributions over $\X$. Then for any $x \in \X$ we have
$|P(x) - Q(x)| = |P(x) - R(x) + R(x) - Q(x)| \leq |P(x) - R(x)| + |R(x) - Q(x)|$.
Summing over all $x\in \X$ yields the inequality.
\item Let $P$ and $P'$ be distributions over $\X$; $Q$ and $Q'$ be distributions over $\Y$.
Then 
\begin{align*}
\VR(P\otimes Q, P' \otimes Q')
	&=\frac{1}{2}\sum_{x\in \X}\sum_{y\in\Y}|P(x)Q(y)-P'(x)Q'(y)|\\
&\leq\frac{1}{2}\sum_{x\in \X}\sum_{y\in\Y}|P(x)Q(y)-P(x)Q'(y)|+|P(x)Q'(y)-P'(x)Q'(y)|\\
&=\frac{1}{2}\sum_{y\in \Y}|Q(y)-Q'(y)|+\frac{1}{2}\sum_{x\in \X}|P(x)-P'(x)| \quad =\VR(P, P') + \VR(Q, Q').
\end{align*}
When $Q=Q'$, the desired result follows
immediately from the first line by factoring out $Q(y)$.
\item Let the distribution of $XY$ be $P(x,y)$ and that of $X'Y'$ be $Q(x,y)$. Let marginals on $\X$ be $P(x)\defeq \sum_y P(x,y)$ and $Q(x) \defeq \sum_y Q(x,y)$. Since $\VR(X,X') = \sum_{x}|P(x)-Q(x)|$, we have
\begin{equation*}
\sum_{x}\big|P(x)-Q(x)\big| = \sum_{x}\Big|\sum_y \big(P(x,y)-Q(x,y)\big)\Big| \leq \sum_{xy} \big|P(x,y)-Q(x,y)\big| =\VR(XY,X'Y'). 
\end{equation*}
\item Let the distribution of $XY$ be $P(x,y)$ and that of $X'Y'$ be $Q(x,y)$. Let marginals on $\X$ be $P(x)\defeq \sum_y P(x,y)$ and $Q(x) \defeq \sum_y Q(x,y)$. Furthermore, let $P(y|x) \defeq P(x,y)/P(x)$ and $Q(y|x)\defeq Q(x,y)/Q(x)$ be the distributions of $Y^x$ and $Y'^x$ respectively.
By assumption, we have $P(x)=Q(x)$. Then we can rewrite $\VR(XY, X'Y') = \frac{1}{2}\sum_{xy}|P(x,y)-Q(x,y)|$ as
\[
\frac{1}{2}\sum_{xy}\big|P(x,y)-Q(x,y)\big|=\frac{1}{2}\sum_x P(x) \sum_y \big|P(y|x)-Q(y|x)\big| = \E_{x\leftarrow X}[\VR(Y^x,Y'^x)].\qedhere
\]
\end{enumerate}
\end{proof}

\subsubsection{Markov chains}

We now define the concept of a Markov chain. We use Markov chains in our analysis because of \fullbref{fact:indep} introduced in \sec{comm}.

\begin{definition}[Markov chain]
We say $XYZ$ is a Markov chain (denoted $X \leftrightarrow Y \leftrightarrow  Z$) if 
\[\Pr(XYZ=(x,y,z)) = \Pr(X=x)\cdot\Pr(Y=y~|~X=x)\cdot\Pr(Z=z~|~Y=y).\]
Equivalently, $XYZ$ is a Markov chain if for every $y$ we have $(XZ)^y \sim X^y \otimes Z^y$.
\end{definition}

The equivalence of the two definitions is shown in \cite[eq.~(2.118), p.~34]{CT06}.
%
%
%
We now present two facts about Markov chains. 

\begin{fact}\label{fact:markov_basic}
	If $X_1X_2YZ_1Z_2$ are random variables and
	$(X_1X_2) \leftrightarrow Y \leftrightarrow  (Z_1Z_2)$,
	then $X_1 \leftrightarrow Y \leftrightarrow  Z_1$.
\end{fact}

\begin{proof}
	Assuming for all $y, X_1^yX_2^yZ_1^yZ_2^y\sim X_1^yX_2^y\otimes Z_1^yZ_2^y$,
	we have
	\begin{align*}
	\Pr(X_1^yZ_1^y=(x_1,z_1))
	&=\sum_{x_2 ,z_2}\Pr(X_1^yX_2^yZ_1^yZ_2^y=(x_1,x_2,z_1,z_2))\\
	&=\sum_{x_2,z_2}\Pr(X_1^yX_2^y=(x_1,x_2))\cdot \Pr(Z_1^yZ_2^y=(z_1,z_2))\\
	&=\Pr(X_1^y=x_1)\cdot\Pr(Z_1^y=z_1).
	\end{align*}
	Thus $(X_1Z_1)^y\sim X_1^y\otimes Z_1^y$.
\end{proof}

\begin{fact}\label{fact:markov-1} 
Let $X \leftrightarrow  Y \leftrightarrow  Z$ be a Markov chain. Then 
\[
\VR(XYZ, X \otimes Y \otimes  Z)  \leq  \VR(XY, X \otimes Y)  +\VR(YZ, Y \otimes Z)  . 
\]
\end{fact} 
\begin{proof}
This follows from the following inequalities.
\begin{align*}
\VR(XYZ, X \otimes Y \otimes  Z) & = \E_{y \leftarrow Y} \VR(X^y \otimes Z^y , X \otimes Z) &  \hspace{-9em}\text{(\fullref{fact:subadd})}\\
&\leq  \E_{y \leftarrow Y} [\VR(X^y \otimes Z^y , X \otimes Z^y) +\VR(X \otimes Z^y , X \otimes Z) ]& \text{(\nameref{fact:triangle})}\\
&=  \E_{y \leftarrow Y} [ \VR(X^y, X )  + \VR(Z^y, Z) ] & \hspace{-9em} \text{(\fullref{fact:prod})} \\
&=  \VR(XY, X \otimes Y)  +\VR(YZ, Y \otimes Z) . &  \hspace{-9em}\text{(\fullref{fact:subadd}) } && \qedhere
\end{align*} 
\end{proof}

\subsubsection{Mutual information}

We now define mutual information and conditional mutual information.

\begin{definition}[Mutual information]
\label{def:entropy}
Let $XYZ \in \X \times \Y \times \Z$ be correlated random variables. We define the following measures, where $\log(\cdot)$ denotes the base $2$ logarithm.
\begin{align*}
\text{Mutual information:}& \quad  \I(X:Y) \defeq \sum_{xy} \Pr(XY=(x,y)) \log \(\frac{\Pr(XY=(x,y))}{\Pr(X=x)\Pr(Y=y)}\) .\\
\text{Conditional mutual information:}& \quad   \I(X:Y~|~Z) \defeq \E_{z\leftarrow Z} \I(X : Y ~|~ Z=z) = \E_{z\leftarrow Z} \I(X^z : Y^z).
\end{align*}

\end{definition}

Mutual information satisfies the following basic properties.

\begin{fact}[Facts about $\I$]
Let $XYZ \in \X \times \Y \times \Z$ be correlated random variables. Then we have the following facts.
\end{fact}
\begin{factpart}[Chain rule]\label{fact:chain-rule}
$\I(X:YZ) = \I(X:Z) + \I(X:Y~|~Z)  = \I(X:Y) + \I(X:Z~|~Y).$
\end{factpart}
\begin{factpart}[Nonnegativity]\label{fact:nonneg}
$\I(X:Y)  \geq 0$ and $\I(X:Y~|~Z)  \geq 0$. 
\end{factpart}
\begin{factpart}[Monotonicity]\label{fact:mono} 
$\I(X:YZ)  \geq  \I(X:Y).$
\end{factpart}
\begin{factpart}[Bar hopping]\label{fact:barhopping}
$\I(X:YZ)\geq \I(X:Y~|~Z)$, where equality holds if $\I(X:Z)=0$.
\end{factpart}
\begin{factpart}[Independence]\label{fact:infoind}
If $Y$ and $Z$ are independent, then $\I(X:YZ) \geq \I(X:Z) + \I(X:Y)$.
\end{factpart}
\begin{factpart}[Data processing]\label{fact:data}
If $X \leftrightarrow Y \leftrightarrow Z$ is a Markov chain, then $\I(X:Y)  \geq  \I(X:Z)$.
\end{factpart}

\begin{proof}
These facts are proved as follows.
\renewcommand{\theenumi}{\bfseries \Alph{enumi}}
\begin{enumerate}
\item See \cite[Theorem $2.5.2$, p.~$24$]{CT06}.
\item See \cite[eq.~(2.90), p.~28]{CT06} and \cite[eq.~(2.92), p.~29]{CT06}.
\item Follows from \fullbref{fact:chain-rule} and \fullbref{fact:nonneg}.
\item From \fullbref{fact:chain-rule} and \fullbref{fact:nonneg}, it follows that 
$\I(X:YZ)=\I(X:Y~|~Z)+\I(X:Z) \geq \I(X:Y~|~Z)$.
\item {Using \fullbref{fact:chain-rule}, we have $\I(X:Z~|~Y) = \I(Z:X~|~Y) = \I(Z:XY) - \I(Z:Y)$. Since $Y$ and $Z$ are independent, $\I(Z:Y)=0$, and hence we get 
\begin{align*}\I(X:Z~|~Y) = \I(Z:XY) \geq \I(Z:X) = \I(X:Z). && \qquad \textrm{(\fullref{fact:mono})}
\end{align*}
Then using \autoref{fact:chain-rule} gives  $\I(X:YZ) = \I(X:Y) + \I(X:Z~|~Y) \geq \I(X:Y)+\I(X:Z)$.}
\item See \cite[Theorem $2.8.1$, p.~34]{CT06}.\qedhere 
\end{enumerate}
\end{proof}

We now present a way to relate mutual information and total variation distance.

\begin{fact}[Relation between $\I$ and $\VR$] \label{fact:ID}
Let $XY \in \X \times \Y$ be correlated random variables. Then
$$\I(X:Y) \geq \VR^2(XY,X\otimes Y) \qquad \textrm{and} \qquad \I(X:Y) \geq \E_{x \leftarrow X} \VR^2(Y^x,Y).$$
\end{fact}
\begin{proof}
To prove this we will require a distance measure called relative entropy (or Kullback--Leibler divergence). For any probability distributions $P$ and $Q$ over $S$, we define
$$\DIV(P \| Q) \defeq \sum_{s \in S} P(s) \log \frac{P(s)}{Q(s)}.$$

We can now express $\I(X:Y)$ in terms of relative entropy as follows:
\begin{equation}
\label{eq:ID}
\I(X:Y) =  \DIV(XY \| X \otimes Y) = \E_{x\leftarrow X} \DIV(Y^x  \| Y).
\end{equation}
The first equality follows straightforwardly from definitions, as shown in \cite[eq.~2.29, p.~20]{CT06}. For the second equality, we proceed as follows:
$$\DIV(XY\| X\otimes Y) = \sum_{x,y}p(x,y)\log\left(\frac{p(x,y)}{p(x)p(y)}\right) = \sum_{x}p(x)\sum_y p(y|x)\log\left(\frac{p(y|x)}{p(y)}\right) = \E_{x\leftarrow X} \DIV(Y^x\| Y).$$

We then use Pinsker's inequality~\cite[Lemma~11.6.1, p.~370]{CT06}, which states
$$\DIV(P \| Q) \geq \frac{2}{\ln 2}\VR^2(P,Q) \geq \VR^2(P,Q).$$
Combining \eq{ID} with Pinsker's inequality completes the proof.
\end{proof}

In general, it is impossible to relate $\I$ and $\VR$ in the
reverse direction. Indeed, mutual information is unbounded, whereas
variation distance is always at most $1$. However, in the case where one of
the variables has binary outcomes, we have the following fact.

\begin{fact}[$\I$ vs.~$\VR$ for binary random variables]\label{fact:I2}
	Let $AB$ be correlated random variables with $A \in \B$. Let $p \defeq \Pr(A=0)$,  $B^0 := (B~|~A=0)$, and $B^1 := (B~|~A=1)$. Then
	\[\I(A:B)\leq  2\log e \enspace  \VR(pB^0,(1-p)B^1).\]
\end{fact}

\begin{proof}
For every $s\in S$, we define 
\[B(s) \defeq pB^0(s) + (1-p)B^1(s) \qquad \textrm{and} \qquad D(s) \defeq pB^0(s)-(1-p)B^1(s),\] 
which gives
$$B^0(s) = \frac{B(s) + D(s)}{2p} \qquad \textrm{and} \qquad B^1(s) = \frac{B(s) - D(s)}{2(1-p)}.$$
Recall that although $\VR(P,Q)$ is a distance measure for probability distributions, it is well defined when $P$ and $Q$ are unnormalized. In particular, $\VR(pB^0,(1-p)B^1)= \frac{1}{2}\sum_s |D(s)|$. 
We can now upper bound $\I(A:B)$ as follows.
\begin{align*}
\I(A:B) &= \sum_{a\in\B} \sum_{s\in S} \Pr(A=a)B^a(s) \log\(\frac{B^a(s)}{B(s)}\)\\
&= \sum_{s\in S} \(pB^0(s) \log\(\frac{B(s)+D(s)}{2pB(s)}\)+(1-p)B^1(s) \log\(\frac{B(s)-D(s)}{2(1-p)B(s)}\)\) \\
&= \mH(p)-1 + \sum_{s\in S} \(pB^0(s) \log\(1+\frac{D(s)}{B(s)}\)+(1-p)B^1(s) \log\(1-\frac{D(s)}{B(s)}\)\),
\end{align*}
where $\mH(p) \defeq -p\log p -(1-p)\log (1-p)$ is the binary entropy function. Since $\mH(p) \leq 1$, we have
\[
\I(A:B) \leq (\log e) \sum_{s\in S} \(pB^0(s) \frac{D(s)}{B(s)}-(1-p)B^1(s) \frac{D(s)}{B(s)}\) = (\log e)\sum_{s\in S} \frac{D(s)^2}{B(s)}, 
\]
using $\log(1+x)\leq x\log e$ (for all real $x$). Since $B^0(s) \geq 0$ and $B^1(s) \geq 0$ for all $s$, we have $|D(s)| \leq B(s)$. Hence
\[
\I(A:B) \leq  (\log e)\sum_{s\in S} \frac{D(s)^2}{B(s)} =(\log e) \sum_s \frac{|D(s)|^2}{B(s)} \leq (\log e) \sum_{s\in S} |D(s)| = 2\log e \enspace \VR(pB^0,(1-p)B^1). \qedhere
\]
\end{proof}

Note that this inequality is tight up to constants. To see this, for any $\delta\in[0,1]$, consider the distributions $B^0=(1-\delta,0,\delta)$ and $B^1=(1-\delta,\delta,0)$. If $p=1/2$, then $I(A:B)=\delta$ and $\VR(pB^0,(1-p)B^1)=\delta/2$.

Our next fact gives us a way to use high mutual information between two variables
to get a good prediction of one variable using a sample from the other.

\begin{fact}[Information $\Rightarrow$ prediction] \label{fact:maxlike}
Let $AB$ be correlated random variables with $A \in \{0,1\}$.
The probability of predicting $A$ by a measurement on $B$ is  at least
\[\frac{1}{2}+\frac{\I(A:B)}{3}.\]
\end{fact}
\begin{proof}
Let $p=\Pr(A=0)$ and define a measurement $M$ corresponding to output $1$ as follows: $M(s)=0$
for all $s \in S$ such that $p B^0(s) \geq  (1-p)B^1(s)$ and $M(s) =1$ otherwise.
We view $M$ as a vector, and let $\mathds{1}$ represents the all-$1$ vector. Then
the success probability of this measurement is 
\begin{align*}
p \langle \mathds{1} -M, B^0\rangle + (1-p) \langle M, B^1\rangle
&= p\langle \mathds{1},B^0\rangle + \langle M,  (1-p)B^1 - p B^0\rangle & \\
&= p+\sum_{s\in S\;:\; (1-p)B^1(s) - p B^0(s)> 0} (1-p)B^1(s) - p B^0(s) & \\
&= p+\frac{1}{2}\sum_{s\in S} |(1-p)B^1(s) - p B^0(s)|+(1-p)B^1(s) - p B^0(s) & \\
&=\frac{1}{2} + \frac{1}{2} \sum_{s\in S}|(1-p) B^1(s)-pB^0(s)| & \\
&=\frac{1}{2} + \VR(pB^0,(1-p)B^1) & 
\end{align*} 
From \fct{I2}, we know that $\VR(pB^0,(1-p)B^1) \geq {\I(A:B)}/(2\log e) \geq {\I(A:B)}/{3}$. 
\end{proof}

\subsection{Communication complexity}
\label{sec:comm}

Let $F\colon \X \times \Y \rightarrow \{0,1,*\}$ be a partial function, with $\dom(F) \defeq \{(x,y)\in \X \times \Y: F(x,y) \neq *\}$, and let $\eps \in (0,1/2)$. In a (randomized) communication protocol for computing   $F$, Alice gets input $x\in \X$, Bob gets input $y \in \Y$. Alice and Bob may use private and public coins. They exchange messages and at the end of the protocol, they output $O(x,y)$. We assume $O(x,y)$ is contained in the messages exchanged by Alice and Bob. We let the random variable $\Pi$ represent the {\em transcript} of the protocol, that is the messages exchanged and the public-coins used in the protocol $\Pi$. Let $\mu$ be a distribution over $\dom(F)$ and let $XY \sim \mu$. We define the following quantities.
\begin{align*}
\textrm{Worst-case error:} & \quad \err(\Pi) \defeq \max_{(x,y)\in \dom(F)} \{ \Pr[O(x,y) \neq F(x,y)] \} .\\
\textrm{Distributional error:} &\quad  \err^\mu(\Pi) \defeq   \E_{(x,y) \leftarrow XY} \Pr[O(x,y) \neq F(x,y)]. \\ 
\textrm{Distributional IC:}&\quad  \IC^\mu(\Pi) \defeq \I(X:\Pi~|~Y) + \I(Y:\Pi~|~X)  . \\
\textrm{Max.\ distributional IC:}&\quad \IC(\Pi)  \defeq \max_{\mu \textrm{ on }\dom(F)}\IC^\mu(\Pi)  . \\
\textrm{IC of $F$:}&\quad  \IC_\eps(F) \defeq \inf_{\Pi: \err(\Pi) \leq \eps} \IC(\Pi) = \inf_{\Pi: \err(\Pi) \leq \eps} \max_{\mu \textrm{ on }\dom(F)}\IC^\mu(\Pi).  \\
\textrm{Randomized CC:}&\quad  \CC(\Pi) \defeq \text{ max.\ number of bits exchanged in $\Pi$ (over inputs and coins)}  . \\
\textrm{Randomized CC of $F$:}&\quad  \R_\eps(F) \defeq \min_{\Pi: \err(\Pi) \leq \eps} \CC(\Pi).  
\end{align*}

Note that since one bit of communication can hold at most one bit of information, for any protocol $\Pi$ and distribution $\mu$ we have $\IC^\mu(\Pi) \leq \CC(\Pi)$. Consequently, we have $\IC_\eps(F) \leq \R_\eps(F)$.
When $\eps$ is unspecified, we assume $\eps = 1/3$. Hence $\IC(F) \defeq \IC_{1/3}(F)$, $\R(F) \defeq R_{1/3}(F)$, and $\IC(F) \leq \R(F)$.  We now present some facts that relate these measures.

Our first fact justifies using $\eps=1/3$ by default since the exact constant does not matter since the success probability of a protocol can be boosted for IC and CC.

\begin{fact}[Error reduction] \label{fact:boost}
Let $0 < \delta < \eps < 1/2$. Let $\Pi$ be a protocol for $F$ with $\err(\Pi) \leq \eps$. There exists protocol $\Pi'$ for $F$ such that $\err(\Pi') \leq \delta$ and   
\begin{align*}
\IC(\Pi') \leq O\left(\frac{\log(1/\delta)}{\big(\frac{1}{2}-\eps\big)^2}\cdot  \IC(\Pi)\right) \quad \textrm{and} \quad 
\CC(\Pi') \leq O\left(\frac{\log(1/\delta)}{\big(\frac{1}{2}-\eps\big)^2}\cdot  \CC(\Pi)\right). 
\end{align*}
\end{fact}

This fact is proved by simply repeating the protocol sufficiently many times and taking the majority vote of the outputs. If the error $\epsilon$ is close to $1/2$, we can first reduce the error to a constant by using $O(\frac{1}{(1/2-\eps)^2})$ repetitions. Then $O\left(\log(1/\delta)\right)$ repetitions suffice to reduce the error down to $\delta$. Hence the communication and information complexities only increase by a factor of $O\left(\frac{\log(1/\delta)}{(1/2-\eps)^2}\right)$.

A useful tool in proving lower bounds on randomized communication complexity is Yao's minimax principle~\cite{Yao77}, which says that the worst-case randomized communication complexity of $F$ is the same as the maximum distributional communication complexity over distributions $\mu$ on $\dom(F)$. In particular, this means there always exists a hard distribution $\mu$ over which any protocol needs as much communication as in the worst case. More precisely, it states that 
$$\R_\eps(F)  =  \max_{\mu \textrm{ on } \dom(F)} \min_{\Pi: \err^\mu(\Pi) \leq \eps} \CC(\Pi).$$

Similar to Yao's minimax principle for randomized communication complexity, we have a (slightly weaker) minimax principle for information complexity due to Braverman~\cite{Bra12}.

\begin{fact}[Minimax principle] \label{fact:equiv}
Let $F\colon\X \times \Y \rightarrow \{0,1,*\}$ be a partial function.
Fix an error parameter $\eps \in (0,1/2)$ and an information bound $I \geq 0$.
Suppose $\mathcal{P}$ is a family of protocols such that for every distribution $\mu$ on $\dom(F)$ there exists a protocol $\Pi \in \mathcal{P}$ such that $$\err^\mu(\Pi) \leq \eps \quad \mathrm{and} \quad \IC^\mu(\Pi)\leq I.$$
Then for any $\alpha \in (0,1)$ there exists a protocol $\Pi'$ such that $$\err(\Pi')\leq \eps/\alpha \quad \mathrm{and} \quad \IC(\Pi') \leq I/(1-\alpha).$$
Moreover, $\Pi'$ is a probability distribution over protocols in $\mathcal{P}$, and hence $\CC(\Pi') \leq \max_{\Pi\in \mathcal{P}} \CC(\Pi)$.
\end{fact}

Our next fact is the observation that if Alice's and Bob's inputs are drawn independently from each other, conditioning on the transcript at any stage of the protocol keeps the input distributions independent of each other.

\begin{fact}[Independence] \label{fact:indep}
Let $\Pi$ be a communication protocol on input $X \otimes Y$.  Then $X\leftrightarrow \Pi \leftrightarrow Y$ forms a Markov chain, or equivalently, for each $\pi$ in the support of $\Pi$, we have
\[(XY)^\pi  \sim X^\pi \otimes Y^\pi .\]
\end{fact}
\begin{proof}
Follows easily by induction on the number of message exchanges in protocol $\Pi$.
\end{proof}

The next property of communication protocols formalizes the intuitive idea that if Alice and Bob had essentially the same transcript for input pairs $(x,y)$ and $(x',y')$, then if we fix Bob's input to either $y$ or $y'$, the transcripts obtained for the two different inputs to Alice are nearly the same. This was shown by  Bar-Yossef, Jayram, Kumar, and Sivakumar~\cite{BJKS04}.

\begin{fact}[Pythagorean property] \label{fact:pyth} 
Let $(x,y)$ and $(x',y')$ be two inputs to a protocol $\Pi$. Then
$$ \h^2(\Pi(x,y), \Pi(x',y)) +  \h^2(\Pi(x,y'), \Pi(x',y'))  \leq 2 \cdot \h^2(\Pi(x,y), \Pi(x',y')).$$
\end{fact}

Our next claim shows that having some information about the output of a Boolean function $F$ allows us to predict the output of $F$ with some probability greater than $1/2$.

\begin{claim} \label{clm:lowinf}
Let $F\colon\X \times \Y \rightarrow \{0,1,*\}$ be a partial function and $\mu$ be a distribution over $\dom(F)$. Let $XY \sim \mu$ and define the random variable $F \defeq F(X,Y)$. Let $\Pi$ be a communication protocol with input $(X,Y)$ to Alice and Bob respectively. There exists a communication protocol  $\Pi'$ for $F$, with input $(X,Y)$ to Alice and Bob respectively, such that 
$$\IC^\mu(\Pi') \leq \IC^\mu(\Pi) + 1,\quad \CC(\Pi') = \CC(\Pi) + 1, \quad\text{and}\quad \err^\mu(\Pi') < \frac{1}{2} - \frac{\I(F:\Pi~|~X)}{3} . $$
\end{claim}
\begin{proof}
In $\Pi'$, Alice and Bob run the protocol $\Pi$ and at the end Alice makes a
prediction for $F$ based on the transcript and her input, essentially
applying \fullbref{fact:maxlike} to the random variables $F^x$ and $\Pi^x$. Alice
then sends her prediction, a single additional bit, to Bob. Clearly,
\[\IC^\mu(\Pi') \leq \IC^\mu(\Pi) + 1 \quad \textrm{and} \quad \CC(\Pi') = \CC(\Pi) + 1 . \]
For every input $x$ for Alice, her prediction is successful (assuming Bob's input
is sampled from $Y^x$) with probability at least ${1}/{2}+{\I(F^x:\Pi^x)}/{3}$
by \fct{maxlike}. Hence the overall success probability of $\Pi'$ is at least
\[\E_{x\leftarrow X}\left[\frac{1}{2}+\frac{\I(F^x:\Pi^x)}{3}\right]
=\frac{1}{2}+\frac{\E_{x\leftarrow X}[\I(F^x:\Pi^x)]}{3}
=\frac{1}{2}+\frac{\I(F:\Pi|X)}{3}.\qedhere \]
\end{proof}

The following claim is used in the proof of our main~\autoref{thm:ICCS} to handle the easy case of a biased input distribution. 

\begin{claim} \label{clm:biasmu}
Let $F\colon\X \times \Y \rightarrow \{0,1,*\}$ be a partial function and let $\mu$ be a distribution over $\dom(F)$.  Let $\epsilon \in (0,1/2)$ and $c \geq 1$ be a positive integer.    For $i \in [c]$, let $X_iY_i \sim \mu$ be i.i.d.\  and define $L_i \defeq F(X_i,Y_i)$. Define $X \defeq (X_1,\ldots,X_c)$, $Y \defeq (Y_1,\ldots,Y_c)$, and $L \defeq (L_1,\ldots,L_c)$. Then 
either
\begin{enumerate}
\item[(a)] There exists a protocol $\Pi$ for $F$ such that $\CC(\Pi) = 1$, $\IC^\mu(\Pi) \leq 1$, and $\err^\mu(\Pi) \leq \frac{1}{2} - \epsilon$, or
\item[(b)] $\VR(XL, X \otimes W^{\otimes c})  \leq c\epsilon $,
where $W$ is the uniform distribution over $\B$.
\end{enumerate}
\end{claim}
\begin{proof}
Define, $q^{x_1} \defeq \Pr[F=0~|~X_1=x_1]$. Assume $\E_{x_1 \leftarrow X_1} [| \frac{1}{2} - q^{x_1}|] \geq \epsilon$.  Let $\Pi$ be a protocol where Alice, on input $x_1$, outputs $0$ if $q^{x_1} \geq 1/2$ and $1$ otherwise. Then, 
$$\err^\mu(\Pi)  =  \frac{1}{2} - \E_{x_1 \leftarrow X_1} [| \frac{1}{2} - q^{x_1}|] \leq \frac{1}{2} - \epsilon. $$
Assume otherwise $\E_{x_1 \leftarrow X_1} [| \frac{1}{2} - q^{x_1}|]  < \epsilon$. Let $W$ be the uniform distribution on $\{0,1\}$. This implies
\begin{align*}
 \VR(XL, X \otimes W^{\otimes c}) &\leq c \cdot \VR(X_1L_1, X_1 \otimes W) = c \cdot  \E_{x_1 \leftarrow X_1} [| \frac{1}{2} - q^{x_1}|] < c\epsilon,
\end{align*}
where the first inequality follows from \fullbref{fact:prod}.
\end{proof}

\section{Lookup functions in communication complexity}
\label{sec:cheat}
We now describe the class of functions we will use for our separations, $(F, \mathcal{G})$-lookup functions.  
This class of communication functions and our applications of them are inspired by the cheat sheet functions defined in 
query complexity in \cite{ABK16}.   

A $(F, \mathcal{G})$-lookup function, denoted $F_\mathcal{G}$, is defined by a (partial) communication function
$F\colon \X \times \Y \to \{0,1,*\}$ and a family $\mathcal{G}=\{G_0, \ldots, G_{2^c-1}\}$ of communication functions, 
where each $G_i \colon (\X^c \times \{0,1\}^m) \times (\Y^c \times \{0,1\}^m) \to \B$.  It can be viewed as a generalization of 
the address function.  Alice receives input $\x=(x_1, \ldots, x_c) \in \X^c$ and $(u_0, \ldots, u_{2^c-1}) \in \{0,1\}^{m 2^c}$ 
and likewise Bob receives input $\y=(y_1, \ldots, y_c) \in \Y^c$ and $(v_0, \ldots, v_{2^c-1}) \in \{0,1\}^{m2^c}$.  The \emph{address}, $\ell$, is determined by the evaluation of $F$ on $(x_1, y_1), \ldots, (x_c,y_c)$, that is 
 $\ell = F^c(\x,\y) \in \{0,1,*\}^c$.  This address (interpreted as an integer in $\{0,\ldots, 2^c-1\}$) then determines which 
 function $G_i$ the  players should evaluate.  If $\ell \in \B^c$, i.e., all $(x_i,y_i) \in \dom(F)$, then the goal of the players is to output 
 $G_\ell(\x, u_\ell, \y, v_\ell)$; otherwise, if some $(x_i,y_i) \not \in \dom(F)$, 
then the goal is to output $G_0(\x, u_0, \y, v_0)$.

The formal definition follows.
\begin{definition}[$(F, \G)$-lookup function]
\label{def:lookup}
Let $F\colon \X \times \Y \to \{0,1,*\}$ be a (partial) communication function and $\mathcal{G}=\{G_0, \ldots, G_{2^c-1}\}$ a 
family of communication functions, where each $G_i \colon (\X^c \times \{0,1\}^m) \times (\Y^c \times \{0,1\}^m) \to \B$.  
A $(F, \mathcal{G})$-lookup function, denoted $F_\mathcal{G}$, is a total communication function 
$F_\G \colon (\X^c \times \B^{m2^c}) \times \Y^c \times \B^{m2^c}$ defined as follows.
Let $\x=(x_1, \ldots, x_c) \in \X^c, \y =(y_1, \ldots, y_c) \in \Y^c, 
\fu=(u_0, \ldots, u_{2^c-1}) \in \B^{m2^c}, \fv=(v_0, \ldots, v_{2^c-1}) \in \B^{m2^c}$.  Then
\[
F_{\G}(\x, \fu, \y, \fv) = 
\begin{cases}
G_\ell(\mathbf{x}, u_\ell, \mathbf{y}, v_\ell) & \text{if } \ell = F^c(\x,\y) \in \B^c \\
G_0(\x, u_0, \y, v_0) & \text{otherwise.}
\end{cases}
\]
\end{definition}

As lookup functions form quite a general class of functions, we will need to impose additional constraints on the 
family of functions $\G$ in order to show interesting theorems about them.  To show
\emph{upper bounds} on the communication complexity of lookup functions (\thm{lookup_upper}), we
need a \emph{consistency} condition.  This says that whenever some $(x_i, y_i) \not \in \dom(F)$, the output of the $G_j$ 
functions can depend only on $\x,\y$ and not on $u,v$ or $j$.
  
\begin{definition}[Consistency outside $F$]
\label{def:consistency}
Let $F\colon \X \times \Y \to \{0,1,*\}$ be a (partial) communication function and $\mathcal{G}=\{G_0, \ldots, G_{2^c-1}\}$ a 
family of communication functions, where each $G_i \colon (\X^c \times \{0,1\}^m) \times (\Y^c \times \{0,1\}^m) \to \B$.  We 
say that $\G$ is \emph{consistent} outside $F$ if for all $i \in \{0, \ldots, 2^c-1\}, u,v,u',v' \in \B^m$ and 
$\x=(x_1, \ldots, x_c) \in \X^c, \y =(y_1, \ldots, y_c) \in \Y^c$ with $\ell=F^c(\x,\y) \not \in \B^c$
we have $G_0(\x, u, \y, v) = G_i(\x, u', \y, v')$.
\end{definition}

In order to show lower bounds on the communication complexity of $F_\G$ (\thm{ICCS}) we add two additional constraints on the family $\G$.  

\begin{definition}[Nontrivial XOR family]
\label{def:xor}
Let $\mathcal{G}=\{G_0, \ldots, G_{2^c-1}\}$ a 
family of communication functions, where each $G_i \colon (\X^c \times \{0,1\}^m) \times (\Y^c \times \{0,1\}^m) \to \B$. We 
say that $\G$ is a nontrivial XOR family if the following conditions hold.
\begin{enumerate}
\item (Nontriviality) For all $\x=(x_1, \ldots, x_c) \in \X^c$ and $\y =(y_1, \ldots, y_c) \in \Y^c$, if we have 
$\ell = F^c(\x,\y) \in \B^c$ 
then for every $i \in \{0,\ldots,2^c-1\}$ there exists $u,v,u', v' \in \B^m$ such that 
$G_i(\x, u, \y, v) \ne G_i(\x, u', \y, v')$.
\item (XOR function) For all $i \in \{0,\ldots,2^c-1\}, u,u',v,v' \in \B^m$ and $\x=(x_1, \ldots, x_c) \in \X^c, 
\y=(y_1, \ldots, y_c) \in \Y^c$ 
if $u \oplus v = u' \oplus v'$ then $G_i(\x, u, \y, v) = G_i(\x, u', \y, v')$.
\end{enumerate}
\end{definition}

\subsection{Upper bound}
We now show a general upper bound on the quantum communication complexity of a $(F, \G)$ lookup function, when 
$\G$ is consistent outside $F$.  A similar result holds for randomized communication complexity, but we will not need this.

\begin{theorem}
\label{thm:lookup_upper}
Let $F\colon \X \times \Y \to \{0,1,*\}$ be a (partial) function and $\mathcal{G}=\{G_0, \ldots, G_{2^c-1}\}$ a 
family of communication functions, where each $G_i \colon (\X^c \times \{0,1\}^m) \times (\Y^c \times \{0,1\}^m) \to \B$.  If 
$\G$ is consistent outside $F$ (\defn{consistency}) then 
\begin{align*}
\Q(F_\G) &=  O(\Q(F) \cdot c \log c) + \max_{i \in [2^c]} O(\Q(G_i)) \\
\Q_E(F_\G) &= \Q_E(F) \cdot c + \max_{i \in [2^c]} \Q_E(G_i)
\end{align*}
where $F_\G$ is the $(F,\G)$-lookup function.
\end{theorem}
\begin{proof}
We first give the proof for the bounded-error quantum communication complexity.

Consider an input where Alice holds $\x=(x_1, \ldots, x_c) \in \X^c$ and 
$\fu= (u_0, \ldots, u_{2^c-1}) \in \{0,1\}^{m 2^c}$ and 
Bob holds $\y=(y_1, \ldots, y_c) \in \Y^c$ and $\fv=(v_0, \ldots, v_{2^c-1}) \in \{0,1\}^{m2^c}$.  For each $i=1, \ldots, c$, 
Alice and Bob run an optimal protocol for $F$ on input $(x_i,y_i)$ $O(\log c)$ many times and let $\ell_i$ be the resulting 
majority vote.  Letting $\ell=(\ell_1, \ldots, \ell_c)$, they then run an optimal protocol for $G_\ell$ on input 
$\x,u_\ell, \y, v_\ell$ a constant number of times and output the majority result.  

The complexity of this protocol is clearly at most $O(\Q(F) \cdot c\log c) + \max_i O(\Q(G_i))$.  We now argue correctness.  
First suppose that each $(x_i,y_i) \in \dom(F)$ for $i=1,\ldots, c$.  In this case, the protocol for $F$ computes 
$F(x_i,y_i)$ with error at most $1/3$.  Thus by running this protocol $O(\log c)$ many times and taking a majority vote
$\ell = (F(x_1, y_1),\ldots, F(x_c,y_c))$ with error probability at most (say) $1/6$.  Similarly by running the protocol for 
$G_\ell$ a constant number of times the error probability can be reduced to $1/6$ and thus the players' output equals
$G_\ell(\x,u_\ell, \y, v_\ell)$ with error probability at most $1/3$.  

If some $(x_i, y_i) \not \in \dom(F)$ then by the consistency condition $G_1(\x,u_1, \y, v_1) = G_\ell(\x, u_\ell, \y, v_\ell)$.  
Thus in this case the players' also output the correct answer with error probability at most $1/3$.

The proof for the exact quantum communication complexity follows similarly.  In this case, Alice and Bob run an 
exact quantum protocol for $F$ on each input $(x_i,y_i)$ to obtain $\ell=(\ell_1, \ldots, \ell_c)$, and then run an 
exact quantum protocol to evaluate $G_\ell$ on input $\x,u_\ell, \y, v_\ell$.  

If each $(x_i,y_i) \in \dom(F)$ for $i=1,\ldots, c$ then $\ell = F(x_1,y_1), \ldots, F(x_c,y_c)$ and the output will be 
correct.  Otherwise, the output will also be correct as $\G$ is consistent outside of $F$.  
\end{proof}

\subsection{Lower bound}
The next theorem is the key result of our work.  It gives a lower bound on the randomized communication complexity 
and information complexity of any $(F,\G)$-lookup function $F_\G$, when $\G$ is a nontrivial XOR family, in terms of the 
same quantities for $F$.  Recall that the value of $F_\G(\x, \fu, \y, \fv)$ is equal to $G_\ell(\x, u_\ell, \y, v_\ell)$, where 
$\ell = F^c(\x,\y)$.  Intuitively, if $\G$ is a nontrivial family, then to evaluate $G_\ell(\x, u_\ell, \y, v_\ell)$ the players must 
at least know the relevant input $u_\ell, v_\ell$.  This in turn requires knowing $\ell$, which can only be figured out
by evaluating $F$.

Since the argument is long, we separate out several claims that will be proven
afterwards. The overall structure of the argument is explained in the main
proof, and displayed visually in \fig{proof_figure}.

	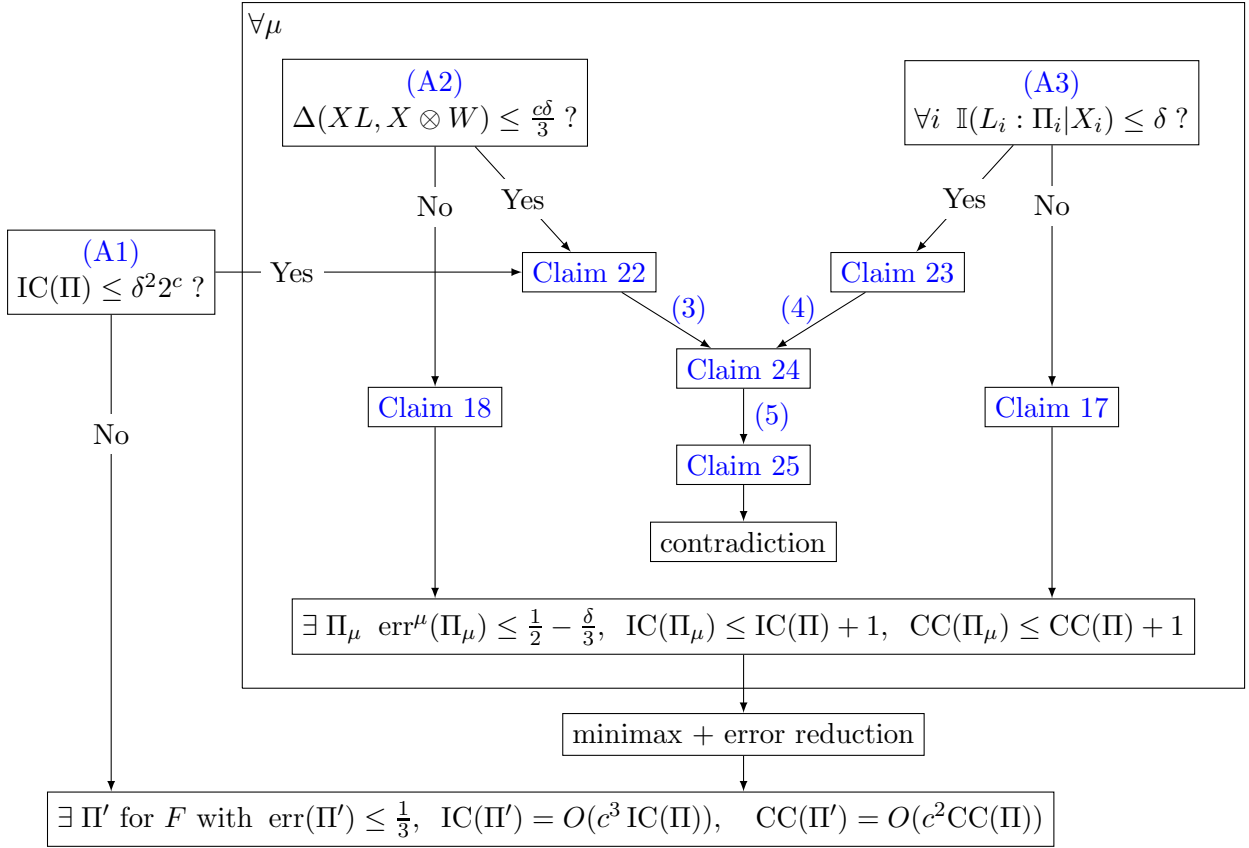
\begin{figure}[tbh]
		\centering
		\begin{tikzpicture}[scale=1.025]
		\pgfmathsetmacro{\left}{0};
		\pgfmathsetmacro{\right}{13};
		\pgfmathsetmacro{\up}{8.4};
		\pgfmathsetmacro{\down}{0.0};
		\pgfmathsetmacro{\middle}{(\left+\right)/2)};
		\draw (\left,\down) -- (\left,\up+0.5) -- (\right,\up+0.5) -- (\right,\down) -- (\left,\down);
		\node (mu) at (\left+0.3,\up+0.2) {$\forall\mu$};
		\node[draw] (conclusion) at (\middle-2.5,\down-1.7) {$\exists\;\Pi'\textrm{ for }F\textrm{ with}\;\;\err(\Pi')\leq\frac{1}{3},\;\;\IC(\Pi')=O(c^3\IC(\Pi)),
			\quad\CC(\Pi')=O(c^2\CC(\Pi))$};
		\node[draw] (minimax) at (\middle, \down-0.6) {minimax + error reduction};
		\node[draw,align=center] (ass1) at (\left-1.7,\up-3) {\eq{assumption1}\\ $\IC(\Pi)\leq \delta^2 2^c\;?$};
		\node[draw,align=center] (ass2) at (\middle-4,\up-0.8) {\eq{assumption2}\\$\VR(XL,X\otimes W)\leq\frac{c\delta}{3}\; ?$};
		\node[draw,align=center] (ass3) at (\middle+4,\up-0.8) {\eq{assumption3}\\$\forall i \;\; \I(L_i:\Pi_i|X_i)\leq\delta\; ?$};
		\node[draw] (X) at (\middle-2,\up-3) {\clm{X}};
		\node[draw] (Y) at (\middle+2,\up-3) {\clm{Y}};
		\node[draw] (pi_mu) at (\middle,\down+0.8) {$\exists\;\Pi_\mu\;\;
			\err^\mu(\Pi_\mu)\leq\frac{1}{2}-\frac{\delta}{3},\;\;\IC(\Pi_\mu)\leq\IC(\Pi)+1,
			\;\;\CC(\Pi_\mu)\leq\CC(\Pi)+1$};
		\node[draw] (Z) at (\middle, \up-4.25) {\clm{Z}};
		\node[draw] (Z2) at (\middle,\up-5.5) {\clm{ZZ}};
		\node[draw] (contradiction) at (\middle, \up-6.5) {contradiction};
		\node[draw] (ass2no) at (\middle-4,\up-4.75) {\clm{biasmu}};
		\node[draw] (ass3no) at (\middle+4,\up-4.75) {\clm{lowinf}};
		\draw[-latex] (ass2) -- (X) node[midway,fill=white] {Yes};
		\draw[-latex] (ass3) -- (Y) node[midway,fill=white] {Yes};
		\draw[-latex] (ass2) -- (ass2no) node[near start,fill=white] {No};
		\draw[-latex] (ass3) -- (ass3no) node[near start,fill=white] {No};
		\draw[-latex] (ass2no.south) -- (ass2.south |- pi_mu.north);
		\draw[-latex] (ass3no.south) -- (ass3.south |- pi_mu.north);
		\draw[-latex] (ass1) -- (X) node[near start,fill=white] {Yes};
		\draw[-latex] (ass1.south) -- (ass1.south |- conclusion.north)
		node[near start,fill=white] {No};
		\draw[-latex] (X) -- (Z) node[near end,above] {\eq{X}};
		\draw[-latex] (Y) -- (Z) node[near end,above] {\eq{Y}};
		\draw[-latex] (Z) -- (Z2)  node[midway,right] {\eq{Z}};
		\draw[-latex] (Z2) -- (contradiction);
		\draw[-latex] (pi_mu) -- (minimax);
		\draw[-latex] (minimax.south) -- (minimax.south |- conclusion.north);
		\end{tikzpicture}
		\caption{The structure of the proof of \thm{ICCS}.
			Note that \clm{X} and \clm{Z} only follow if both of their incoming
			arcs hold.}
		\label{fig:proof_figure}
	\end{figure}

In \thm{ICCS}, we are given a bounded-error protocol $\Pi$ for $F_\G$, and our goal is to construct a bounded-error protocol $\Pi'$ for $F$ such that its communication complexity and information complexity do not increase by more than a polynomial in $c$ compared to the protocol $\Pi$.

As depicted in \fig{proof_figure}, if \eq{assumption1} fails to hold, then we are done. Otherwise, we assume $\mu$ is an arbitrary distribution over $\dom(F)$, and check if \eq{assumption2} or \eq{assumption3} hold. We show that it is not possible for both to hold, since that leads to a contradiction. If either \eq{assumption2} or \eq{assumption3} fail to hold, then we have a protocol $\Pi_\mu$ that does well for the distrubition $\mu$. Finally we apply a minimax argument, which converts protocols that work well against a known distribution into a protocol that works on all inputs, and obtain the desired protocol $\Pi'$.

\subsubsection{Main result}

\begin{theorem}
\label{thm:ICCS}
Let $F\colon\X \times \Y \to \{0,1,*\}$ be a (partial) function and let $c\geq \log \R(F)$. Let $\G=\{G_0, \ldots, G_{2^c-1}\}$ 
be a nontrivial family of XOR functions (\defn{xor}) where each $G_i \colon (\X^c \times \{0,1\}^m) \times (\Y^c \times \{0,1\}^m) \to \B$, and 
let $F_\G$ be the $(F, \G)$-lookup function.  For any $1/3$-error protocol $\Pi$ for $F_\G$, there exists a $1/3$-error protocol $\Pi'$ for $F$ such that 
$$\IC(\Pi') \leq O(c^3\IC(\Pi)) \quad \text{and} \quad  \CC(\Pi') \leq O(c^2 \CC(\Pi)).$$
In particular,
$\R(F_\G)= \Omega(\R(F)/c^2)$ and $\IC(F_\G)= \Omega(\IC(F)/c^3)$.
\end{theorem}

\begin{proof}
In this proof, for convenience we define $\delta = \frac{1}{10^{22}c}$ (we are not trying to optimize the constants).

\para{Rule out trivial protocols.} We first rule out the easy case where the protocol we are given, $\Pi$, has very high information complexity. More precisely, we check if the following condition holds.
\begin{equation*}
\label{eq:assumption1}
\IC(\Pi) < \delta^2 2^c
\tag{A1}
\end{equation*}
If this does not hold then $\IC(\Pi) \geq \delta^2 2^c = \Omega(\R(F)/c^2)$. By choosing the protocol whose communication complexity is $\R(F)$, we obtain a protocol $\Pi'$ for $F$ with $\IC(\Pi') \leq \CC(\Pi') = \R(F) = O(c^2\IC(\Pi))$ and we are done. Hence for the rest of the proof we may assume \eq{assumption1}.

\para{Protocols correct on a distribution.}
Instead of directly constructing a protocol $\Pi'$ for $F$ that is correct on all inputs with bounded error, we instead construct for every distribution $\mu$ on $\dom(F)$, a protocol $\Pi_\mu$ that does well on $\mu$ and then use \fullbref{fact:equiv} to construct our final protocol. More precisely, for every $\mu$ over $\dom(F)$  we construct a protocol $\Pi_\mu$ for $F$ that has the following properties:
\begin{align}
\IC^\mu(\Pi_\mu)  \leq \IC(\Pi) + 1, \qquad \CC(\Pi_\mu)  = \CC(\Pi) + 1 \qquad \textrm{and} \qquad \err^\mu(\Pi_\mu) < 1/2-\delta/3. \label{eq:niceprot}
\end{align}

Hence for the remainder of the proof let $\mu$ be any distribution over $\dom(F)$ and our aim is to construct a protocol satisfying \eq{niceprot}.

\para{Construct a distribution for $F_\G$.}
Using the distribution $\mu$ on $\dom(F)$, we now construct a distribution over the inputs to $F_\G$.
Let the random variable $T$ be defined as follows: 
\[
T \defeq(X_1,\ldots,X_c,U_0,\ldots,U_{2^c-1},Y_1\ldots,Y_c,V_0,\ldots,V_{2^c-1}),
\]
where for all $ i \in [c]$, $X_iY_i$ is distributed according to $\mu$ and is independent of all other random variables and for  $j \in \{0, \ldots, 2^c-1\}$, $U_jV_j$ are uniformly distributed in $\{0,1\}^{2m}$ and independent of all other random variables. 

For  $i \in [c]$, we define  $L_i \defeq F(X_i,Y_i)$. We also define $X \defeq (X_1,\ldots,X_c)$, $Y \defeq (Y_1,\ldots,Y_c)$, $L \defeq (L_1,\ldots,L_c)$, $U \defeq (U_0,\ldots,U_{2^c-1})$, and $V \defeq (V_0,\ldots,V_{2^c-1})$. Lastly, for $i \in [c]$ we define $X_{-i} \defeq X_1 \ldots X_{i-1} X_{i+1} \ldots X_c$ and  $X_{<i} \defeq X_1 \ldots X_{i-1}$. Similar definitions hold for $L$ and $Y$.

\para{Rule out easy distributions $\mu$.} We now show that if $\mu$ is such that the output of $F(X,Y)$ is predictable simply by looking at Alice's input $X$, then this distribution is easy and we can construct a protocol $\Pi_\mu$ that does well on this distribution since Alice can simply guess the value of $F(X,Y)$ after seeing $X$. More precisely, we check if the following condition holds.
\begin{equation*}
\label{eq:assumption2}
\VR(XL, X \otimes W)  \leq c \delta/3,
\tag{A2}
\end{equation*}
where $W$ is the uniform distribution on $\{0,1\}^c$. 

If the condition does not hold, we invoke \clm{biasmu} with $\epsilon =\delta/3$. Then we must be in case~(a) of this claim and hence we get the desired protocol $\Pi_\mu$.
Therefore we can assume \eq{assumption2} holds.

\para{Construct new protocols $\Pi_i$.}
We now define a collection of protocols $\Pi_i$ for each $i\in [c]$.  $\Pi_i$ is a protocol in which Alice and Bob receive inputs from $\dom(F)$. We construct $\Pi_i$ as follows: Given the input pair $(X_i,Y_i)$ distributed according to $\mu$, Alice and Bob use their public coins to sample $c-1$ inputs $X_{-i}Y_{-i}$ according to $\mu^{\otimes c}$ and inputs $U$ and $V$ uniformly at random. Note that Alice and Bob now have inputs $XU$ and $YV$ distributed according to $T$. The random variable corresponding to their transcript, which includes the messages exchanges and the public coins, is  $\Pi_i = (\Pi, X_{-i}, U, Y_{-i}, V)$.   We then claim that 
\[
\forall i \in [c], \quad \CC(\Pi_i) = \CC(\Pi) \quad \text{ and } \quad \IC^\mu(\Pi_i) \leq \IC(\Pi).
\]
It is obvious that  $\CC(\Pi_i) = \CC(\Pi)$ because the bits transmitted in $\Pi_i$ are the same as those in $\Pi$. The second part uses the following chain of inequalities, which hold for any $i\in [c]$.
\begin{align*}
\IC(\Pi) \geq \IC^T(\Pi)  &= \I(X U: \Pi~|~ Y V) + \I(YV: \Pi~|~ XU) &\quad \textrm{(definition)}\\
& \geq  \I(X_i : \Pi~|~ Y_i X_{-i}  U Y_{-i} V) + \I(Y_i : \Pi~|~ X_i X_{-i} U Y_{-i} V) & \,\,\, \textrm{(\fullref{fact:barhopping})}\\
& = \I(X_i : \Pi X_{-i}  U Y_{-i} V ~|~ Y_i ) + \I(Y_i : \Pi X_{-i} U Y_{-i} V~|~ X_i )  & \,\,\, \textrm{(\fullref{fact:barhopping})}\\
& = \IC^\mu(\Pi_i). &\quad \textrm{(definition)}
\end{align*}
The second equality used the fact that $ \I(X_i : X_{-i}  U Y_{-i} V~|~ Y_i )  =  \I(Y_i : X_{-i}  U Y_{-i} V~|~ X_i ) = 0$.

\para{Rule out informative protocols $\Pi_i$.}
We then check if any of the protocols $\Pi_i$ that we just constructed have a lot of information about the output $L_i$. If this happens then $\Pi_i$ can solve $F$ on $\mu$ and will yield the desired protocol $\Pi_\mu$. More precisely, we check if the following condition holds.
\begin{equation*}
\label{eq:assumption3}
\forall i \in [c] \quad \I(L_i : \Pi_i ~|~ X_i) \leq  \delta.
\tag{A3}
\end{equation*}
If it does not hold, then we apply \clm{lowinf}, which gives us the desired protocol $\Pi_\mu$ satisfying \eq{niceprot}. Hence we may assume that \eq{assumption3} holds for the rest of the proof.

\para{Obtain a contradiction.}
We have already established that \eq{assumption1}, \eq{assumption2}, and \eq{assumption3} must hold, otherwise we have obtained our protocol $\Pi_\mu$. We will now show that if  \eq{assumption1}, \eq{assumption2}, and \eq{assumption3}  simultaneously hold, then we obtain a contradiction. To show this, we use some claims that are proved after this theorem.

First we apply \clm{X} to get the following from \eq{assumption1} and \eq{assumption2}.
\begin{equation}
\Pr_{(x,\ell) \leftarrow XL} (\VR((\Pi U_\ell)^{x},\Pi^x  \otimes U_\ell  ) > \sqrt{\delta} )  < 0.01.  \label{eq:X}
\end{equation}
Intuitively this claim asserts that for a typical $x$ and $\ell$, the transcript $\Pi^x$ has very little information about the correct cell $U_\ell$, which is quantified by saying their joint distribution is close to being a product distribution. This would be false without assuming \eq{assumption1} because if there was no upper bound on the information contained in $\Pi$, then the protocol could simply communicate all of $U$, in which case it would have a lot of information about any $U_j$. We need \eq{assumption2} as well, since otherwise it is possible that the correct answer $\ell$ is easily predicted by Alice by looking at her input alone, in which case she can send over the contents of cell $U_\ell$ to Bob.

We then apply \clm{Y} to get the following from \eq{assumption3}.
\begin{equation}
\Pr_{(x,\ell) \leftarrow XL} (\VR((\Pi U_\ell)^{x,\ell} , (\Pi U_\ell)^x  ) > 100  \sqrt{c \delta})\leq 0.01 \label{eq:Y}
\end{equation}
Intuitively, this claim asserts that for a typical $x$, the transcript (and even the
contents of $U_\ell$, Alice's part of the contents of the correct cheat sheet cell)
does not change much upon further conditioning on $\ell$. This is just one way
of saying that Alice (who knows $x$ and $U$) does not learn much about $\ell$
from the transcript $\Pi$. The assumption \eq{assumption3} was necessary,
because without it, it would be possible for $\Pi$ to provide a lot of information
about $L$ (conditioned on $X$).

We then apply \clm{Z}, which uses \eq{X} and \eq{Y} to obtain the following:
\begin{align} 
\Pr_{(x,y,\ell,u_\ell,v_\ell) \leftarrow XYLU_LV_L} \left(\VR(\Pi^{x,y,\ell,u_\ell,v_\ell}  ,  \Pi^{x,y,\ell} ) > 6 \cdot 10^6 \cdot \sqrt{c\delta}  \right) < 0.09   \label{eq:Z}.
\end{align}

This equation is a key result. It says that conditioning on a typical $(x,y,\ell)$, the message transcript does not change much on further conditioning on a typical $(u_\ell, v_\ell)$.
Finally, we use \clm{ZZ} to obtain a contradiction from \eq{Z}. 

\para{Minimax argument.}
Note that in all branches where we did not reach a contradiction, we constructed a protocol satisfying \eq{niceprot}. Hence we constructed, for any $\mu$ over $\dom(F)$, a protocol $\Pi_\mu$ that satisfies \eq{niceprot}. 
From here it is easy to complete the proof. First we use \fullbref{fact:equiv} with the choice $\alpha = 1-\frac{\delta}{6}$ and $\eps = \frac{1}{2}-\frac{\delta}{3}$ to get a protocol $\widetilde{\Pi}$ for $F$ such that 
\[ \IC(\widetilde{\Pi}) \leq O\left(\frac{1}{\delta}\IC(\Pi)  \right), \qquad  \CC(\widetilde{\Pi}) \leq \CC(\Pi) +1 \qquad \text{and} \qquad \err(\widetilde{\Pi}) \leq 1/2-\delta/6.\]
Finally, using \fullbref{fact:boost}, we get a protocol $\Pi'$ for $F$ such that
$$\IC(\Pi') \leq O\left(\frac{1}{\delta^3} \IC(\Pi)  \right), \qquad \CC(\Pi') \leq O\left(\frac{1}{\delta^2}\CC(\Pi)  \right) \qquad \text{and} \qquad \err(\Pi') \leq 1/3.$$
This completes the proof since $1/\delta = O(c)$. 
\end{proof}

This completes the proof of the theorem, except the claims we did not prove, \clm{X}, \clm{Y}, \clm{Z}, and \clm{ZZ}. We now prove these claims.

\subsubsection{Proofs of claims}

\begin{claim}\label{clm:X}
Assume the following conditions hold.
\begin{equation*}
\IC(\Pi) < \delta^2 2^c \tag{\ref{eq:assumption1}}
\end{equation*}
\begin{equation*}
\VR(XL, X \otimes W)  \leq c \delta/3 \tag{\ref{eq:assumption2}}
\end{equation*}
Then we have 
\begin{equation*}
\Pr_{(x,\ell) \leftarrow XL} (\VR((\Pi U_\ell)^{x},\Pi^x  \otimes U_\ell  ) > \sqrt{\delta} )  < 0.01.  \tag{\ref{eq:X}}
\end{equation*}
\end{claim}

\begin{proof}
Using \eq{assumption1}, we have
\begin{align}
\delta^2 2^c > \IC(\Pi) &>  \IC^T(\Pi)  = \I(UX:\Pi~|~YV) + \I(YV:\Pi~|~XU)& \textrm{(definition)} \nonumber \\
&\geq  \I(UX:\Pi~|~YV) & \text{(\fullref{fact:nonneg})} \nonumber \\
&\geq  \I(U:\Pi ~|~ XYV) & \textrm{(\fullref{fact:barhopping})}\nonumber  \\
&=  \I(U:\Pi YV ~|~ X) & \textrm{(\fullref{fact:barhopping})}\nonumber  \\
&\geq  \I(U:\Pi ~|~ X)  & \text{(\fullref{fact:mono})} \nonumber \\
&=   \E_{x \leftarrow X}   \I(U:\Pi ~|~ X=x) & \hspace{-6em}\text{(\fullref{def:entropy})} \nonumber  \\
&=   \E_{x \leftarrow X}   \I(U_1^x\cdots U_{2^c}^x:\Pi^x) &  \text{(notation)} \nonumber  \\
&\geq  \E_{x \leftarrow X}\textstyle \sum_{\ell=1}^{2^c}  \I(U_\ell^x:\Pi^x) & \hspace{-1em} \text{(\fullref{fact:infoind})} \nonumber  \\
&=  2^c \, \E_{x \leftarrow X} \E_{\ell \leftarrow W}  \I(U_\ell^x:\Pi^x). &  \text{($W$ is the uniform distribution)} \nonumber  \\
\Rightarrow \quad \delta^2 & >  \E_{(x,\ell) \leftarrow X \otimes W}  \I(U_\ell^x  : \Pi^x)   .\nonumber  \\
\Rightarrow \quad \delta  & > \Pr_{(x,\ell) \leftarrow X \otimes W}(\I(U_\ell^x  : \Pi^x) > \delta)  .  & \text{(\fullref{fact:markov_ineq})} \nonumber 
\end{align}
We now want to replace the distribution $X\otimes W$ with $XL$ on the right hand side. Since  $\VR(XL,X\otimes W) \leq c\delta/3$ from \eq{assumption2}, changing the distribution from $X\otimes W$ to $XL$ only changes the probability of any event by $2c\delta/3\leq c\delta$. Therefore 
\begin{align}
0.01 > c\delta +\delta   &> \Pr_{(x,\ell) \leftarrow XL}(\I(U_\ell^x  : \Pi^x) > \delta)   & \nonumber \\
& \geq \Pr_{(x,\ell) \leftarrow XL} (\VR^2((\Pi U_\ell)^{x},\Pi^x  \otimes U_\ell^x  ) > \delta) & \text{(\fullref{fact:ID})}  \nonumber  \\
& =  \Pr_{(x,\ell) \leftarrow XL} (\VR^2((\Pi U_\ell)^{x},\Pi^x  \otimes U_\ell) > \delta) & \text{($U_\ell$ is independent of $X$)}  \nonumber  \\
& = \Pr_{(x,\ell) \leftarrow XL} (\VR((\Pi U_\ell)^{x},\Pi^x  \otimes U_\ell  ) > \sqrt{\delta} )   .  \label{eq:l3} &\qedhere
\end{align}
\end{proof}

\begin{claim}\label{clm:Y}
Assume the following condition holds.
\begin{equation*}
\forall i \in [c] \quad \I(L_i : \Pi_i ~|~ X_i) \leq  \delta\tag{\ref{eq:assumption3}}
\end{equation*}
Then we have 
\begin{equation*}
 \Pr_{(x,\ell) \leftarrow XL} (\VR((\Pi U_\ell)^{x,\ell} , (\Pi U_\ell)^x  ) > 100  \sqrt{c \delta})\leq 0.01.  \tag{\ref{eq:Y}}
\end{equation*}
\end{claim}

\begin{proof}
We first show that $(\Pi U)$ together carries low information about $L$ even conditioned on $X$.  More precisely we show that
\begin{equation}\label{eq:lowinfL}
c \delta \geq  \I(L:\Pi U~|~X).
\end{equation}
This follows from the following chain on inequalities starting with \eq{assumption3}.
\begin{align*}
\delta & \geq \I(L_i : \Pi_i ~|~ X_i) \\ 
& = \I(L_i : \Pi X_{-i} U Y_{-i} V ~|~ X_i) &\quad \textrm{(definition of $\Pi_i$)}\\ 
& = \I(L_i : \Pi X_{-i} U Y_{-i} V X_{<i}Y_{<i} ~|~ X_i)  &\quad \textrm{($X_{<i}Y_{<i}$ contained in $X_{-i}Y_{-i})$}\\ 
&\geq \I(L_i: X_{<i}Y_{<i} \Pi  U  ~|~ X) &\quad \textrm{(\fullref{fact:barhopping} and \fullref{fact:mono})}\\ 
&\geq \I(L_i: L_{<i} \Pi  U ~|~ X)  & \quad \text{(\fullref{fact:data})} \\ 
& =  \I(L_i:\Pi U ~|~ L_{<i} X). &\quad \textrm{(\fullref{fact:barhopping})} \\
 \intertext{By summing this inequality over $i$, we get}
c \delta 
& \geq\textstyle \sum_{i=1}^c \I(L_i:\Pi U ~|~ L_{<i} X) & \\
& =  \I(L:\Pi U ~|~  X). & \textrm{(\fullref{fact:chain-rule})} 
\end{align*}
This is \eq{lowinfL}, which we set out to show. Using this inequality, we have
\begin{align}
c \delta & \geq  \I(L:\Pi U ~|~  X)     \nonumber \\
& =  \E_{x \leftarrow X} \I(L: \Pi U ~|~ X=x)  &  \text{(\fullref{def:entropy})}   \nonumber \\
& =  \E_{x \leftarrow X} \I(L^x: (\Pi U)^x)  &  \text{(notation)}   \nonumber \\
& \geq \E_{(x,\ell) \leftarrow XL} \VR^2((\Pi U)^{x,\ell}, (\Pi U)^x) & \hspace{-1em}\text{(\fullref{fact:ID})}  \nonumber  \\
& \geq \E_{(x,\ell) \leftarrow XL} \VR^2((\Pi U_\ell )^{x,\ell} , (\Pi U_\ell)^x ).  & \text{(\fullref{fact:monoTV})} \nonumber  \\
\Rightarrow \quad  \sqrt{c \delta} & \geq \E_{(x,\ell) \leftarrow XL} \VR((\Pi U_\ell)^{x,\ell} , (\Pi U_\ell)^x  ).   & \text{(convexity of square)}  \nonumber  \\
\Rightarrow \quad  0.01  & \geq \Pr_{(x,\ell) \leftarrow XL} (\VR((\Pi U_\ell)^{x,\ell} , (\Pi U_\ell)^x  ) > 100  \sqrt{c \delta} ) &  \text{(\fullref{fact:markov_ineq})} \nonumber \tag{\ref{eq:Y}} 
\end{align}
This completes the proof.
\end{proof}

\begin{claim}\label{clm:Z}
Assume the following conditions hold.
\begin{equation*}
\Pr_{(x,\ell) \leftarrow XL} (\VR((\Pi U_\ell)^{x},\Pi^x  \otimes U_\ell  ) > \sqrt{\delta} )  < 0.01. \tag{\ref{eq:X}}
\end{equation*}
\begin{equation*}
 \Pr_{(x,\ell) \leftarrow XL} (\VR((\Pi U_\ell)^{x,\ell} , (\Pi U_\ell)^x  ) > 100  \sqrt{c \delta})\leq 0.01. \tag{\ref{eq:Y}}
\end{equation*}
Then we have 
\begin{equation*}
\Pr_{(x,y,\ell,u_\ell,v_\ell) \leftarrow XYLU_LV_L} \left(\VR(\Pi^{x,y,\ell,u_\ell,v_\ell}  ,  \Pi^{x,y,\ell} ) > 6 \cdot 10^6 \cdot \sqrt{c\delta}  \right) < 0.09.  \tag{\ref{eq:Z}}
\end{equation*}
\end{claim}

\begin{proof}
First, using \eq{Y} we can show
\begin{align}
0.01  & \geq \Pr_{(x,\ell) \leftarrow XL} (\VR((\Pi U_\ell)^{x,\ell} , (\Pi U_\ell)^x  ) > 100  \sqrt{c \delta} ) &   \tag{\ref{eq:Y}} \nonumber \\
 & \geq \Pr_{(x,\ell) \leftarrow XL} (\VR(\Pi^{x,\ell} , \Pi^x)  > 100  \sqrt{c \delta} )   &  \qquad \text{(\fullref{fact:monoTV})}  \nonumber \\
 & = \Pr_{(x,\ell) \leftarrow XL} (\VR(\Pi^{x,\ell} \otimes U_\ell, \Pi^x\otimes U_\ell)  > 100  \sqrt{c \delta} ). & \qquad \text{(\fullref{fact:prod})}   \label{eq:YY}
\end{align}

Using \eq{X}, \eq{Y}, and \eq{YY}, the union bound and \fullbref{fact:triangle} we get
\begin{align} 
0.03 &>   \Pr_{(x,\ell) \leftarrow XL} \left(\VR((\Pi U_\ell)^{x,\ell}  ,  \Pi^{x,\ell}  \otimes U_\ell ) > 300 \sqrt{c\delta}\right)  & \nonumber \\
&=\Pr_{(x,\ell) \leftarrow XL}
\left(\E_{\pi\leftarrow\Pi^{x,\ell}}(\VR(U_\ell^{x,\ell,\pi} , U_\ell ) )>
300 \sqrt{c\delta}\right) & \hspace{-5em} \text{(\fullref{fact:subadd})}\nonumber \\
&=\Pr_{(x,\ell) \leftarrow XL}
\left(\E_{\pi\leftarrow\Pi^{x,\ell}}(\VR(U_\ell^{x,\ell,\pi}\otimes Y^{x,\ell,\pi} , U_\ell\otimes Y^{x,\ell,\pi} ) )>
300 \sqrt{c\delta}\right) & \text{(\fct{prod})}\nonumber \\
&=\Pr_{(x,\ell) \leftarrow XL}
\left(\E_{\pi\leftarrow\Pi^{x,\ell}}(\VR(U_\ell^{x,\ell,\pi}Y^{x,\ell,\pi} , U_\ell\otimes Y^{x,\ell,\pi} ) )>
300 \sqrt{c\delta}\right) & \text{(\fullref{fact:indep})}\nonumber \\
&=\Pr_{(x,\ell) \leftarrow XL}
\left(\VR(U_\ell^{x,\ell}Y^{x,\ell}\Pi^{x,\ell} , U_\ell\otimes Y^{x,\ell}\Pi^{x,\ell} )>
300 \sqrt{c\delta}\right) &  \hspace{-5em} \text{(\fullref{fact:subadd})}\nonumber \\
&=\Pr_{(x,\ell) \leftarrow XL}
\left(\E_{y\leftarrow Y^{x,\ell}}(\VR(U_\ell^{x,y,\ell}\Pi^{x,y,\ell} , U_\ell\otimes \Pi^{x,y,\ell} ) ) >
300 \sqrt{c\delta}\right) &\text{(\fct{subadd})}\nonumber
%
%
%
\end{align}
where the third equality follows since for all $(x,\ell)$, the variables $(U_\ell \Pi Y)^{x,\ell}$ form a Markov chain. To see this, fix $x$ and $\ell$,
and consider giving Alice the input $x$ together with an input distributed from
$U^{x,\ell}$. Also, give Bob an input generated from $(YV)^{x,\ell}$.
Since $U$ is uniform and independent of everything else, Alice's input is independent
of Bob's. \fullbref{fact:indep} then implies that
$U^{x,\ell}\leftrightarrow\Pi^{x,\ell}\leftrightarrow(YV)^{x,\ell}$
is a Markov chain. Then \fct{markov_basic} allows us to conclude
$U_\ell^{x,\ell}\leftrightarrow \Pi^{x,\ell}\leftrightarrow Y^{x,\ell}$.

Next, using \fullbref{fact:markov_ineq}, we get
\begin{align}
0.04 & > \Pr_{(x,y,\ell) \leftarrow XYL}\left(\VR((U_\ell \Pi )^{x,y,\ell}  ,  U_\ell   \otimes (\Pi^{x,y,\ell})  ) > 30000 \sqrt{c\delta}\right).  \label{eq:l4}
\end{align}
By symmetry between Alice and Bob, we get
\begin{align}
0.04 & > \Pr_{(x,y,\ell) \leftarrow XYL}\left(\VR((V_\ell \Pi )^{x,y,\ell}  ,  V_\ell   \otimes (\Pi^{x,y,\ell})  ) > 30000 \sqrt{c\delta}\right)  .  \label{eq:l5}
\end{align}
Using Eqs.~\eqref{eq:l4} and~\eqref{eq:l5}  and the union bound we get
\begin{align*} 0.08& > \Pr_{(x,y,\ell) \leftarrow XYL} \left( \VR((U_\ell \Pi )^{x,y,\ell}  ,  U_\ell   \otimes (\Pi^{x,y,\ell})  )  +  \VR((V_\ell \Pi )^{x,y,\ell}  ,  V_\ell   \otimes (\Pi^{x,y,\ell})  ) > 60000\sqrt{c\delta}  \right) \hspace{-2em}&  \\
& \geq   \Pr_{(x,y,\ell) \leftarrow XYL} \left(\VR((U_\ell \Pi V_\ell )^{x,y,\ell}  ,  U_\ell   \otimes (\Pi^{x,y,\ell})  \otimes V_\ell )  > 60000\sqrt{c\delta}  \right), & \text{(\fct{markov-1})}
\end{align*}
where the last inequality used the fact that $(U_\ell \Pi V_\ell )^{x,y,\ell}$ is a Markov chain, which follows from a similar argument to before.
Using \fullbref{fact:subadd} and \fullbref{fact:markov_ineq}, we then get
\[
0.09 >   \Pr_{(x,y,\ell,u_\ell,v_\ell) \leftarrow XYLU_LV_L} \left(\VR(\Pi^{x,y,\ell,u_\ell,v_\ell}  ,  \Pi^{x,y,\ell} ) > 6 \cdot 10^6 \cdot \sqrt{c\delta}  \right). \qedhere
\] 
\end{proof}

\begin{claim}\label{clm:ZZ}
If we assume
\begin{equation}
\Pr_{(x,y,\ell,u_\ell,v_\ell) \leftarrow XYLU_LV_L} \left(\VR(\Pi^{x,y,\ell,u_\ell,v_\ell}  ,  \Pi^{x,y,\ell} ) > 6 \cdot 10^6 \cdot \sqrt{c\delta}  \right) < 0.09  \tag{\ref{eq:Z}}
\end{equation}
then we have a contradiction.
\end{claim}

\begin{proof}
Eq.~\eq{Z} implies that there exists $(x,y,\ell)$ such that 
\begin{align*} 
& 0.09>   \Pr_{(u_\ell,v_\ell) \leftarrow U_\ell V_\ell} \left(\VR(\Pi^{x,y,\ell,u_\ell,v_\ell}  ,  \Pi^{x,y,\ell} ) > 6 \cdot 10^6 \cdot \sqrt{c\delta}  \right)   .
\end{align*}

Recall that $G_\ell$ only depends on the XOR of the strings $u_\ell$ and $v_\ell$.  We assume without loss of 
generality that the number of 
strings $s \in \B^m$ such that $G_\ell(x,u_\ell, y, v_\ell)=1$ when $u_\ell \oplus v_\ell=s$ is at least the number of 
strings $s$ for 
which $G_\ell(x,u_\ell, y, v_\ell)=0$ when $u_\ell \oplus v_\ell=s$.  A symmetric argument holds otherwise.
This implies that 
there exists a string $s$ such that $G_\ell(x,u_\ell, y, v_\ell)=1$ whenever $u_\ell \oplus v_\ell = s$ and
\[
 0.18>   \Pr_{u_\ell \leftarrow U_\ell} \left(\VR(\Pi^{x,y,\ell,u_\ell,u_\ell \oplus s}  ,  \Pi^{x,y,\ell} ) > 6 \cdot 10^6 \cdot  \sqrt{c\delta}  \right).
\]
Fix any $t \in \B^m$ such that $G_\ell(x,u_\ell, y, v_\ell)=0$ whenever $u_\ell \oplus v_\ell=t$. The inequality above implies 
that there exists a string $u_\ell$  such that
\begin{align*} 
6 \cdot 10^6 \cdot  \sqrt{c\delta} & \geq   \VR(\Pi^{x,y,\ell,u_\ell,u_\ell \oplus s}  ,  \Pi^{x,y,\ell} )   \\
 \text{and} \quad  6 \cdot 10^6 \cdot  \sqrt{c\delta} & \geq   \VR(\Pi^{x,y,\ell,u_\ell \oplus t \oplus s,u_\ell \oplus  t}  ,  \Pi^{x,y,\ell} )    .
\end{align*}
Using \fullbref{fact:triangle} we get
\begin{align*} 
0.001\geq 12 \cdot 10^6 \cdot  \sqrt{c\delta} & \geq   \VR(\Pi^{x,y,\ell,u_\ell,u_\ell \oplus s}  , \Pi^{x,y,\ell,u_\ell \oplus s \oplus t,u_\ell \oplus t} ) \\
 &\geq \h^2(\Pi^{x,y,\ell,u_\ell,u_\ell \oplus s}  , \Pi^{x,y,\ell,u_\ell \oplus s \oplus t,u_\ell \oplus t} )    &  \hspace{-1em} \text{(\fullref{fact:relation-inf})}  \\
&=\h^2((\Pi^{x,y,\ell})^{u_\ell,u_\ell \oplus s}  , (\Pi^{x,y,\ell})^{u_\ell \oplus s \oplus t,u_\ell \oplus t} )    &  \text{(notation)}\\
 & \geq  \frac{1}{2} \h^2((\Pi^{x,y,\ell})^{u_\ell,u_\ell \oplus s}  , (\Pi^{x,y,\ell})^{u_\ell \oplus s \oplus t, u_\ell \oplus s } ) & \hspace{-1em} \text{(\fullref{fact:pyth})} \\
 & \geq  \frac{1}{4}\VR^2(\Pi^{x,y,\ell,u_\ell,u_\ell \oplus s}  , \Pi^{x,y,\ell,u_\ell \oplus s \oplus t, u_\ell \oplus s } )   & \hspace{-1em} \text{(\fullref{fact:relation-inf})} \\
\Rightarrow 0.1> \sqrt{0.004} &>  \VR(\Pi^{x,y,\ell,u_\ell,u_\ell \oplus s}  , \Pi^{x,y,\ell,u_\ell \oplus s \oplus t, u_\ell \oplus s } )   .
\end{align*}
This implies that the worst case error of protocol $\Pi$ is at least $0.5 - 0.1 > 1/3$. This is a contradiction because $\Pi$ was assumed to have error less than $1/3$.
\end{proof}


\section{Randomized lower bounds for \autoref{thm:q-vs-r} and \autoref{thm:qe-vs-r}}
\label{sec:R_lower}
In this section we will prove the randomized communication complexity lower bounds needed for the separation against 
bounded-error quantum communication complexity of \autoref{thm:q-vs-r} and the separation against exact quantum 
communication complexity separation of 
\autoref{thm:qe-vs-r}.  We start by giving a high-level overview of the whole proof, before showing the randomized 
lower bounds in this section, and the quantum upper bounds in the next section.

In both cases, we begin in the world of query complexity.  The starting point of \autoref{thm:q-vs-r} is
the partial function 
\begin{equation}
\label{eq:str}
\STR \coloneqq \SIMON_n \circ \OR_n \circ \AND_n.
\end{equation}
Here $\SIMON_n$ is a certain property testing version of Simon's problem \cite{Sim97} introduced in \cite{BFNR08} 
(defined in \autoref{sec:simon} below) which witnesses a large gap between its randomized $\Rdt(\SIMON_n) = 
\Omega(\sqrt{n})$ and quantum $\Qdt(\SIMON_n) = O(\log n \log \log n)$ query complexities.  
As shown in~\cite[\S3]{ABK16}, the 
cheat sheet 
version of  $\STR$ witnesses an $\tO(n)$-vs-$\tOmega(n^{2.5})$ separation between quantum and randomized query 
complexities.  (Actually, they use $\FORR$~\cite{AA15} in place of $\SIMON$, but we find it more convenient to work with 
$\SIMON$.)  

We follow a similar approach to the query case and first ``lift'' $\STR$ to a partial two-party function 
$F=\STR \circ \IP_b$ by 
composing it with $\IP_b$, the two-party inner-product function on $b=\Theta(\log n)$ bits per party.  Our final function 
achieving the desired separation will be a $(F, \G)$-lookup function $F_\G$ where $\G$ forms a consistent family of 
nontrivial XOR functions, described in \sec{quantum-ub}.  

By~\autoref{thm:ICCS}, to show a lower bound on the 
randomized communication complexity of $F_\G$, it suffices to show a randomized communication lower bound on 
$F=\STR \circ \IP_b$.  To do this, we use the query-to-communication lifting theorem of \cite{GLM+15}, 
which requires us to show a lower bound on the approximate \emph{conical junta degree} of $\STR$ (see \autoref{sec:juntadegree} for definitions). For this, we would like to show that each of 
$\SIMON_n$, 
$\OR_n$, $\AND_n$ individually have large junta degree and then invoke a \emph{composition theorem} for conical junta 
degree~\cite{GJ16}.  Because of certain technical conditions in the composition theorem, we will actually need to show 
a lower bound on the functions $\SIMON_n$, 
$\OR_n$, $\AND_n$ in a slightly stronger model, giving dual certificates for these functions of a special form.  This will 
prove \autoref{thm:lifted-lb}. 

The other half of \autoref{thm:q-vs-r} is a quantum upper bound on the communication complexity of $F_\G$, for a 
particular family of functions $\G$.  We 
need that the family $\G$ is consistent 
outside $F$, and that each function $G_i \in G$ 
has $\Q(G_i) = \tO(n)$.  We do this in a way very analogous to the cheat sheet framework: each function 
$G_i(\x,u,\y,v)$ evaluates to $1$ if and only if $u \oplus v$ verifies that $(x_i, y_i) \in \dom(F)$ for all $i \in [c]$.  The 
players check this using a distributed version of Grover search.  The formal definition of $F_\G$ and the upper bound 
on its quantum communication complexity appear in \autoref{sec:quantum-ub}.

For the separation between randomized and exact quantum communication complexity, we begin in the setting of 
query complexity with the partial function
\begin{equation}
\label{eq:ptr}
\pTR_{n,m} \coloneqq \pOR_n \circ \AND_m,
\end{equation}
where we eventually choose $m = \Theta(\sqrt{n})$ and $\pOR_n$ is a promise version of the $\OR_n$ function
\[
\pOR_n(x) = 
\begin{cases}
0 & \text{ if } |x| = 0 \\
1 & \text{ if } |x| = 1 \\
* & \text{ otherwise}
\end{cases} \enspace .
\]
The exact quantum query complexity of $\pTR$ is $O(\sqrt{n}m)$, while its randomized query complexity is 
$\Omega(nm)$.  As shown in~\cite[\S6.4]{ABK16}, taking $m = \Theta(\sqrt{n})$, the 
cheat sheet version of $\pTR$ is a total function that witnesses an $\tO(n)$ versus $\Omega(n^{3/2})$ separation 
between randomized and exact quantum query complexities.  

We again lift $\pTR$ to a partial two-party function $H \coloneqq \pTR \circ \IP_b$ by composing it with $\IP_b$.%
\footnote{For this separation one could alternatively use $\pOR_n \circ \Disj_m$, where 
$\Disj_m(\x,\y) \coloneqq \wedge_{i=1}^m \neg x_i \vee \neg y_i$.  Jayram et al.~\cite{JKS03} show an $\Omega(mn)$ randomized lower bound for $\OR\circ\Disj_m$ and a closer look at their proof, especially at the key lemma of \cite{BJKS04} that is used (which we reproduce in \autoref{fact:AND}), shows that the argument also works for the promise version. However, to give a unified exposition, we prefer to work with $\pTR \circ \IP_b$ here.}
The final function for the separation of \autoref{thm:qe-vs-r} will be a $(H,\T)$-lookup function for a 
particular family of XOR functions $\T$ that consistent outside of $H$ and defined in a similar fashion to the 
family $\G$ described above.

The main theorem of this section is the following.
\begin{theorem} \label{thm:lifted-lb}
Let $m \le n$ and let $b \ge t \log n$ for a sufficiently large constant $t$.  Then 
\begin{align*}
\R(\STR\circ \IP_b) = \tOmega(n^{2.5}) \qquad \textrm{and} \qquad 
\R(\pTR_{n,m} \circ \IP_b) = \Omega(nm).
\end{align*}
\end{theorem}
The plan for both of these lower bounds is similar, as outlined in Figure~\ref{fig:lower}.   
\begin{figure}[tbh]
        \label{fig:lower}
	\centering 
	\begin{tikzpicture}[scale=1.025]
	\pgfmathsetmacro{\left}{0};
	\pgfmathsetmacro{\right}{13};
	\pgfmathsetmacro{\up}{9.4};
	\pgfmathsetmacro{\down}{0.0};
	\pgfmathsetmacro{\middle}{(\left+\right)/2)};
	\pgfmathsetmacro{\sep}{2.8};
	\draw (\left,\down+.5) rectangle (\right,\up+1.2);
	\node (simon) at (\middle-3.25,\up) [draw, align=center]{\autoref{lem:simon_cert} \\ one-sided $(\sqrt{n}/2,0)$ junta \\ 
	certificate for $\SIMON_n$};
	\node (or) at (\middle+3.25,\up) [draw, align=center]{\autoref{lem:or_cert}\\two-sided $(n/2,n)$ junta \\ certificate for $\pOR_n$};
	\node[draw, align=center] (STR) at (\middle-3.25,\up-\sep) {$\deg_\eps(\STR) = \Omega(n^{2.5})$ \\ $\STR \coloneqq \SIMON_n \circ \OR_n \circ \AND_n$};
	\node[draw, align=center] (PTR) at (\middle+3.25,\up-\sep) {$\deg_\epsilon(\pTR) = \Omega(nm)$\\$\pTR_{n,m} \coloneqq \pOR_n \circ \AND_m$};
	\node[draw, align=center] (F) at (\middle-3.25,\up-2*\sep) {$R_\epsilon(\STR \circ \IP_b) = \Omega(n^{2.5})$\\$F \coloneqq \STR \circ \IP_b$};
	\node[draw, align=center] (H) at (\middle+3.25,\up-2*\sep) {$R_\epsilon(\pTR_{n,m} \circ \IP_b) = \Omega(nm)$\\$H \coloneqq \pTR_{n,m} \circ \IP_b$};
	\node[draw] (FG) at (\middle-3.25,\up-2.9*\sep) {$R_\epsilon(F_\G) = \tOmega(n^{2.5})$};
	\node[draw] (HG) at (\middle+3.25,\up-2.9*\sep) {$R_\epsilon(H_\G) = \tOmega(nm)$};
	\node[align=center] (lift) at (\middle,\up-1.5*\sep) {\autoref{thm:simu} \\ (Lifting theorem \cite{GLM+15})};
	\node[align=center] (lift) at (\middle,\up-2.5*\sep) {\autoref{thm:ICCS}\\(Lookup lower bound)};
	\draw [->] (simon) -- (STR);
	\draw [->] (or) -- (STR) node[midway,fill=white, align=center]{\autoref{thm:composition}\\ (Composition theorem \cite{GJ16})};
	\draw [->] (or) -- (PTR);
	\draw [->] (STR) -- (F);
	\draw [->] (PTR) -- (H);
	\draw [->] (F) -- (FG);
	\draw [->] (H) -- (HG);	
	\end{tikzpicture}
	\caption{Overview of the randomized communication complexity lower bounds for \autoref{thm:q-vs-r} and \autoref{thm:qe-vs-r}.}
\end{figure}
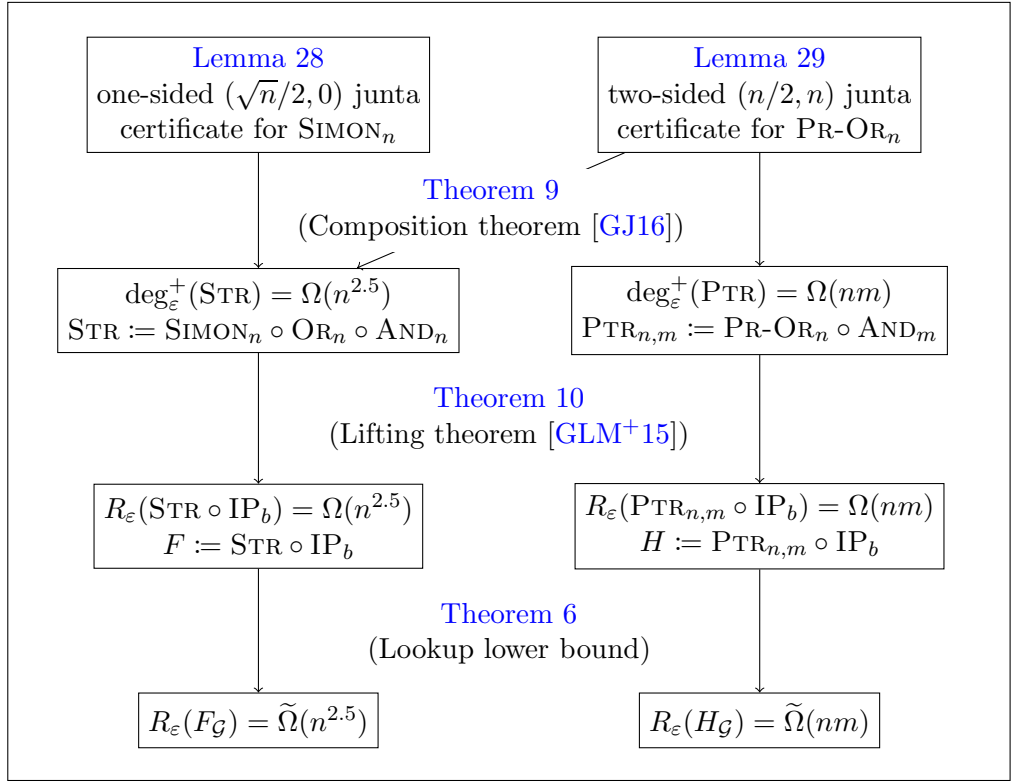
Following this outline, our first task in proving \autoref{thm:lifted-lb} is to give junta certificates for the component 
functions $\SIMON_n, \OR_n, \pOR_n, \AND_n$ that make up $\STR$ and $\pTR$.  This is done in the next subsection.

\subsection{Conical junta degree}
\label{sec:juntadegree}
A \emph{conical junta} $h$ is a nonnegative linear combination of conjunctions; more precisely, $h=\sum_C w_C C$ where $w_C\geq 0$ and the sum ranges over all conjunctions $C\colon\{0,1\}^n\to\{0,1\}$ of literals (input bits or their negations). For a conjunction $C$ we let $|C|$ denote its width, i.e., the number of literals in $C$.  The \emph{conical 
junta degree} of $h$, denoted $\deg(h)$, is the maximum width of a conjunction $C$ with $w_C>0$. Any conical junta $h$ naturally computes a nonnegative function $h\colon\{0,1\}^n\to\mathbb{R}_{\geq 0}$.  For a partial boolean function $f\colon\{0,1\}^n\to\{0,1,*\}$ we say that $h$ \emph{$\epsilon$-approximates} $f$ if and only if $|f(x) - h(x)|\leq \epsilon$ for all inputs $x\in\dom(f)$. The \emph{$\epsilon$-approximate conical junta degree of $f$}, denoted $\deg_\epsilon(f)$, is defined as the minimum degree of a conical junta $h$ that $\epsilon$-approximates $f$. (Note: $\deg_\epsilon(f)$ is also known as the query complexity analogue of the one-sided smooth rectangle bound~\cite{JK10}.)

In this subsection we establish the following lower bound.
\begin{theorem} \label{thm:junta_deg}
Let $\STR$ and $\pTR_{n,m}$ denote the partial functions defined in \eq{str} and \eq{ptr}. Then
\begin{align*}
\deg_{1/(64\sqrt{n})}(\STR) = \Omega(n^{2.5}) \qquad \textrm{and} \qquad 
\deg_{1/16}(\neg \pTR_{n,m}) = \Omega(nm).
\end{align*}
\end{theorem}

To prove \autoref{thm:junta_deg} we use a composition theorem for conical junta degree due to \cite{GJ16}.  
For example, for the first statement we would ideally like to conclude the $\Omega(n^{2.5})$ lower bound from the fact that $\SIMON_n$, $\OR_n$, $\AND_n$ (and some of their negations) have approximate conical junta degrees 
$\Omega(\sqrt{n})$, $\Omega(n)$, $\Omega(n)$, respectively. These facts are indeed implicit in existing literature; for example:
\begin{itemize}[itemsep=-2pt]
\item The result of~\cite{BFNR08} recorded in \autoref{lem:simon_balanced} implies $\deg_{1/3}(\SIMON_n)\geq\Omega(\sqrt{n})$.
\item Klauck~\cite{Klauck10} has proved even a communication analogue of $\deg_\eps(\OR_n)\geq\Omega(n)$. (For an exposition of the query version, see, e.g., \cite[\S4.1]{GJPW15}.)
\item Jain and Klauck~\cite[\S3.3]{JK10} proved that $\deg_{1/16}(\OR_n \circ \AND_n)\geq\Omega(n^2)$.
\end{itemize}
Unfortunately, the composition theorem from \cite{GJ16} assumes some regularity conditions from the \emph{dual certificates} witnessing these lower bounds. 
(In fact, without regularity assumptions, a composition theorem for a related ``average conical 
junta degree'' measure is known to fail! See~\cite[\S3]{GJ16} for a discussion.) 
We now review the composition theorem before constructing the special dual certificates in 
Section~\ref{sec:simon} and~\ref{sec:or_cert}.

\paragraph{Composition theorem.}
We recall the necessary definitions from \cite{GJ16} in order to state the composition theorem precisely. The theorem was originally phrased for total functions, but the result holds more generally for partial functions $f$ provided the dual certificates are supported on the domain of $f$. The following definitions make these provisions.

A function $\Psi\colon\{0,1\}^n\to\mathbb{R}$ is \emph{balanced} if $\sum_x \Psi(x)=0$. Write $X_{\geq 0}\coloneqq\max\{0,X\}$ for short. A \emph{two-sided $(\alpha,\beta)$ junta certificate} for a partial function $f\colon\{0,1\}^n\to\{0,1,*\}$ consists of four balanced functions $\{\Psi_v,\hat{\Psi}_v\}_{v=0,1}$ satisfying the following:
\begin{itemize}[noitemsep]
\item \emph{Special form:} There exist distributions $D_1$ over $f^{-1}(1)$ and $D_0$ over $f^{-1}(0)$ such that $\Psi_1 = \alpha\cdot(D_1 - D_0)$ and $\Psi_0 = -\Psi_1$. Moreover, $\hat{\Psi}_v$ is supported on $f^{-1}(v)$.
\item \emph{Bounded 1-norm:} For each $v\in\{0,1\}$ we have $\|\hat{\Psi}_v\|_1\leq \beta$.
\item \emph{Feasibility:} For all conjunctions $C$ and $v\in\{0,1\}$,
\begin{equation}\label{eq:feasibility}
\langle \Psi_v,C\rangle_{\geq 0} + \langle\hat{\Psi}_v,C\rangle\leq|C|\langle D_v,C\rangle.
\end{equation}
\end{itemize}
We also define a \emph{one-sided $(\alpha,\beta)$ junta certificate} for $f$ as a pair of balanced functions $\{\Psi_1,\hat{\Psi}_1\}$ that satisfies the above conditions but only for $v=1$.

\begin{theorem}[Composition theorem~\cite{GJ16}] \label{thm:composition}
Suppose $f:\{0,1,*\}^n \rightarrow \{0,1\}$ admits a two-sided (resp.~one-sided) $(\alpha_1,\beta_1)$ junta certificate, and $g$ admits a two-sided $(\alpha_2,\beta_2)$-junta certificate. Then $f\circ g$ admits a two-sided (resp.~one-sided) $(\alpha_1\alpha_2,\beta_1+n\beta_2)$ junta certificate.
\end{theorem}
\begin{lemma}[Junta degree lower bounds from certificates~\cite{GJ16}] \label{lem:lb-from-certs}
Suppose $f$ admits a one-sided $(\alpha,\beta)$ junta certificate. Then $\deg_\epsilon(f)\geq\Omega(\alpha)$ provided $\epsilon<1/4$ and $\epsilon\beta\leq \alpha/4$.
\end{lemma}

\subsubsection{Junta certificate for \texorpdfstring{$\SIMON_n$}{SIMON}}
\label{sec:simon}
The partial function $\SIMON_n\colon\{0,1\}^n\to\{0,1,*\}$ is defined as follows. (For convenience, we actually use the negation of the function defined in~\cite{BFNR08}.) We interpret the input $z\in\{0,1\}^n$ as a function 
$z\colon \mathbb{Z}_2^d\to\{0,1\}$ where $d=\log n$ (we tacitly assume that $n$ is a power of 2, which can be achieved by padding). Call a function $z$ \emph{periodic} if there is some nonzero $y\in \mathbb{Z}_2^d$ such that $z(x+y)=z(x)$ for all $x\in \mathbb{Z}_2^d$. Furthermore, $z$ is \emph{far from periodic} if the Hamming distance between $z$ and every periodic function is at least $n/8$. Then
\[
\SIMON_n(z)\coloneqq
\begin{cases}
1 & \text{if $z$ is far from periodic,} \\
0 & \text{if $z$ is periodic,} \\
* & \text{otherwise.}
\end{cases}
\]
The key properties of this function, proved in~\cite[\S4]{BFNR08}, are:
\begin{itemize}[noitemsep]
\item \emph{Quantum query complexity:} $\Qdt(\SIMON_n)\leq O(\log n\log\log n)$.
\item \emph{Randomized query complexity:} $\Rdt(\SIMON_n)\geq \Omega(\sqrt{n})$.
\end{itemize}
Moreover, it is important for us that the randomized lower bound is robust: it is witnessed by a pair of distributions $(D_1,D_0)$ where $D_i$ is supported on $(\SIMON_n)^{-1}(i)$ such that any small-width conjunction that accepts under $D_1$ also accepts under $D_0$ with comparable probability. We formalize this property in the following lemma; for completeness, we present its proof (which is implicit in~\cite[\S4]{BFNR08}). One subtlety is that the property is \emph{one-sided} in that the statement becomes false if we switch the roles of $D_1$ and $D_0$.
\begin{lemma}[\cite{BFNR08}]
\label{lem:simon_balanced} 
Let $\alpha\coloneqq\sqrt{n}/2$. There exists a pair of distributions $(D_1,D_0)$ where $D_i$ is supported on $\SIMON_n^{-1}(i)$ such that for every conjunction $C$ with $|C|$ literals,
\begin{equation} \label{eq:simon-lb}
\Pr_{z \leftarrow D_0}[C(z)=1] \geq (1- |C|/\alpha)\cdot\Pr_{z \leftarrow D_1}[C(z)=1].
\end{equation}
\end{lemma}

\begin{proof}
Assume $1\leq|C|\leq\alpha$ for otherwise the claim is trivial. Define $U$ and $D_1$ as the uniform distributions on $\{0,1\}^n$ and $\SIMON_n^{-1}(1)$, respectively. Define a distribution $D_0$ on periodic functions $z$ as follows: choose a nonzero period $y \in \mathbb{Z}_2^d$ uniformly at random, and for every edge of the matching $(x, x + y)$ in 
$\mathbb{Z}_2^d$ uniformly choose $b \in \B$ and set $b=z(x)=z(x+y)$.

Intuitively, $C$ can distinguish between $z \leftarrow D_0$ and a uniformly random string only if $C$ queries two input vectors whose difference is the hidden period $y$ that was used to generate $z$. Indeed, let $\mathcal{S}\subseteq\mathbb{Z}_2^d$, $|\mathcal{S}|\leq \binom{|C|}{2}$, be the set of vectors of the form $x+x'$ where $C$ queries both $z(x)$ and $z(x')$. Then, conditioned on the event ``$y\notin \mathcal{S}$'', the bits $C$ reads from $z$ are uniformly random. Hence
\begin{align} \label{eq:d0}
\Pr_{z\leftarrow D_0}[C(z)=1]
&\geq \Pr_{z\leftarrow D_0}[y\notin \mathcal{S}  \wedge  C(z)=1] \notag \\
&= \Pr_{z\leftarrow D_0}[y\notin \mathcal{S}]\cdot\Pr_{z\sim D_0}[C(z)=1\mid y\notin \mathcal{S}] \notag \\
&\geq\textstyle \big(1-\binom{|C|}{2}/(n-1)\big)\cdot 2^{-|C|} \notag \\
&\geq\textstyle \big(1-|C|/\sqrt{n}\big)\cdot 2^{-|C|},
\end{align}
where the last inequality holds because $|C|\leq \alpha$.

Since there are at most $n2^{n/2}$ periodic functions, there are at most $n2^{n/2}\cdot 2^{n H(1/8)} \leq 2^{2n/3}$ functions at Hamming distance $\leq n/8$ from periodic functions (here $H$ is the binary entropy function). Hence the total variation distance between $U$ and $D_1$ is tiny: $\VR(U,D_1) \leq 2^{-\Omega(n)}$. Thus
\begin{equation} \label{eq:d1}
\Pr_{z\leftarrow D_1}[C(z)=1]
\leq 2^{-|C|} + 2^{-\Omega(n)} \leq (1+2^{-\Omega(n)})\cdot 2^{-|C|}.
\end{equation}
Putting \eqref{eq:d0} and \eqref{eq:d1} together we get
\begin{align*}
\Pr_{z\leftarrow D_0}[C(z)=1]
&\geq \big(1-|C|/\sqrt{n})(1-2^{-\Omega(n)})\cdot \Pr_{z\leftarrow D_1}[C(z)=1]\\
&\geq \big(1-2|C|/\sqrt{n})\cdot \Pr_{z \leftarrow D_1}[C(z)=1]. \qedhere
\end{align*}
\end{proof}

With this lemma in hand, we now show that $\SIMON_n$ has the needed junta certificate.
\begin{lemma}
\label{lem:simon_cert}
$\SIMON_n$ has a one-sided $(\sqrt{n}/2,0)$ junta certificate.
\end{lemma}

\begin{proof}
A one-sided $(\alpha,0)$ junta certificate, $\alpha\coloneqq\sqrt{n}/2$, is given by
\[
\arraycolsep=2pt
\begin{array}{rclp{1.5cm}rcl}
\Psi_1 &\coloneqq& \alpha\cdot(D_1-D_0), && \hat{\Psi}_1 &\coloneqq& 0, \\
\end{array}
\]
where $(D_1,D_0)$ are from~\autoref{lem:simon_balanced}. Note that \eq{simon-lb} can be rephrased as $\langle D_0,C\rangle\geq (1-|C|/\alpha)\langle D_1,C\rangle$ since $\langle D_v,C\rangle = \Pr_{z\sim D_v}[D(z)=1]$. The feasibility condition \eqref{eq:feasibility} follows:
\begin{align*} \label{eq:simon-calc}
\langle\Psi_1,C\rangle_{\geq 0}+\langle\hat{\Psi}_1,C\rangle 
&= \langle \Psi_1,C\rangle \\
&= \alpha\langle D_1,C\rangle - \alpha\langle D_0,C\rangle \\
&\leq \alpha\langle D_1,C\rangle - \alpha(1-|C|/\alpha)\langle D_1,C\rangle \\
&= |C|\langle D_1,C\rangle.
\end{align*}
\end{proof}

\subsubsection{Junta certificates for \texorpdfstring{$\OR_n,\AND_n$ and $\pOR_n$}{OR, AND and PR-OR}}
\label{sec:or_cert}
\begin{lemma}
\label{lem:or_cert}
$\OR_n$ and $\AND_n$ have two-sided $(n/2,n)$ junta certificates.  The negation of $\pOR_n$ 
has a one-sided $(n/2,0)$ junta certificate.
\end{lemma}
\begin{proof}
We show that for $\OR_n$, a two-sided $(n/2,n)$ junta certificate is given by
\begin{equation}
\label{eq:or_cert}
\arraycolsep=2pt
\begin{array}{rclp{1.5cm}rcl}
\Psi_1 &\coloneqq& n/2\cdot(D_1-D_0), && \hat{\Psi}_1 &\coloneqq& n/2\cdot(D_1-D_2), \\
\Psi_0 &\coloneqq& n/2\cdot(D_0-D_1), && \hat{\Psi}_0 &\coloneqq& 0,
\end{array}
\end{equation}
where $D_i$ is the uniform distribution on inputs of Hamming weight $i$.  As $\Psi_0, \hat{\Psi}_0$ only have 
support on inputs of Hamming weight zero or one, this will also imply that the negation of $\pOR_n$ has a 
one-sided $(n/2,0)$ junta certificate.  By duality of $\OR_n$ and $\AND_n$ it will also imply that $\AND_n$ has a 
two-sided $(n/2,n)$ junta certificate.  Thus we focus on showing that \autoref{eq:or_cert} forms a valid certificate.

To check the feasibility conditions~\eqref{eq:feasibility}, we split into cases depending on how many positive literals $C$ 
contains. For notation, let $C_j$ be a conjunction of width $w\coloneqq|C|$ having $j$ positive literals (and thus $w-j$ 
negative literals). We have the following table of acceptance probabilities:
\begin{center}
\renewcommand{\arraystretch}{1.2}
\vspace{3mm}
\begin{tabular}{r|@{\hspace{1em}}c@{\hspace{1em}}c@{\hspace{1em}}c}
\toprule
$j$ & $\langle D_0,C_j\rangle$ & $\langle D_1,C_j\rangle$ & $\langle D_2,C_j\rangle$  \\
\midrule
$0$ & $1$ & $(n-w)/n$ & $\binom{n-w}{2}/\binom{n}{2}$ \\
$1$ & $0$ & $1/n$ & $(n-w)/\binom{n}{2}$ \\
$2$ & $0$ & $0$ & $1/\binom{n}{2}$ \\
$\ge 3$ & $0$ & $0$ & $0$ \\
\bottomrule
\end{tabular}
\vspace{3mm}
\end{center}
For $v=1$, it suffices to consider $j\in\{0,1\}$ since any $C_j$ with $j>1$ will have $\langle D_1,C_j\rangle =~ 0$ and hence $\langle \Psi_1,C_j\rangle, \langle \hat{\Psi}_1,C_j\rangle \leq 0$. For $v=~0$, it suffices to consider $j=~0$ since any $C_j$ with $j>0$ will have $\langle D_0,C_j\rangle =~ 0$ and hence $\langle \Psi_0,C_j\rangle \leq 0$. Here we go:
\begin{align*}
\textstyle
\langle\Psi_1,C_0\rangle_{\geq 0} + \langle\hat{\Psi}_1,C_0\rangle
~=~&\textstyle 0 + n/2\cdot\langle D_1-D_2,C_0\rangle\\
=~&\textstyle n/2\cdot\big(\frac{n-w}{n}-\binom{n-w}{2}/\binom{n}{2}\big)
=\textstyle n/2\cdot\big(\frac{n-w}{n}(1-\frac{n-w-1}{n-1})\big) \\
=~&\textstyle n/2\cdot\big(\frac{n-w}{n}\cdot\frac{w}{n-1}\big)
=\textstyle 1/2\cdot(n-w)\cdot\frac{w}{n-1} \\
\leq~&\textstyle (n-w)\cdot\frac{w}{n} = w\langle D_1,C_0\rangle,\\[3mm]
\langle\Psi_1,C_1\rangle_{\geq 0} + \langle\hat{\Psi}_1,C_1\rangle
~=~&\textstyle n/2\cdot \langle D_1,C_1\rangle + n/2\cdot \langle D_1-D_2,C_1\rangle\\
=~&\textstyle n \cdot \langle D_1,C_1\rangle - n/2\cdot \langle D_2,C_1\rangle\\
=~&\textstyle 1- n/2\cdot (n-w)/\binom{n}{2}
=\textstyle 1- \frac{n-w}{n-1}
=\textstyle \frac{w-1}{n-1} \\
\leq~&\textstyle w/n = w\langle D_1,C_1\rangle,\\[3mm]
\langle\Psi_0,C_0\rangle_{\geq0} + \langle\hat{\Psi}_0,C_0\rangle 
~=~&\textstyle \langle\Psi_0,C_0\rangle
=\textstyle n/2\cdot \langle D_0-D_1,C_0\rangle \\
=~&\textstyle n/2\cdot (1-(n-w)/n)
=\textstyle w/2 \\
\leq~&\textstyle w = w\langle D_0,C_0\rangle.
\end{align*}
\end{proof}

\subsubsection{Junta degree lower bounds}
With the composition theorem and the junta certificates just constructed, we can now prove 
\autoref{thm:junta_deg}.

\begin{proof}[Proof of \autoref{thm:junta_deg}]
To show that $\deg_\epsilon(\STR) = \Omega(n^{2.5})$ we begin with the 
two-sided $(n/2,n)$ certificates for $\OR_n, \AND_n$ from \autoref{lem:or_cert}.  
By \autoref{thm:composition} this gives a two-sided $(n^2/4, n^2 +n)$ junta certificate for $\OR_n \circ \AND_n$.  
Now composing with the one-sided $(\sqrt{n}/2,0)$ junta certificate for $\SIMON_n$ gives a 
one-sided $(n^{2.5}/8,n^3+n^2)$ junta certificate for $\STR$.  Finally applying \autoref{lem:lb-from-certs} gives 
$\deg_\epsilon(\STR) = \Omega(n^{2.5})$ for any $\epsilon \le (64\sqrt{n})^{-1}$.

The negation of $\pOR_n$ has a one-sided $(n/2,0)$ junta certificate and $\AND_m$ has a two-sided $(m/2,m)$ 
certificate.  Thus by \autoref{thm:composition}, $\neg \pTR_{n,m}$ has a one-sided $(nm/4, nm)$ junta certificate.  
Applying \autoref{lem:lb-from-certs} shows that $\deg_\epsilon(\neg \pTR_{n,m}) = \Omega(mn)$ for any 
$\epsilon \le 1/16$.
\end{proof}

\subsection{From query to communication}

The following theorem is a corollary of \cite[Theorem 31]{GLM+15} (originally, the theorem was stated for constant $\epsilon>0$, but the theorem holds more generally for $\epsilon = 2^{-\Theta(b)}$; note also that instead of $\deg_\epsilon(f)$, that paper uses the notation~$\textrm{WAPP}^\textrm{dt}_\epsilon(f)$). Here $\IP_b\colon\{0,1\}^b\times\{0,1\}^b\to\{0,1\}$ is the two-party inner-product function defined by $\IP_b(x,y)\coloneqq \langle x,y\rangle \bmod{2}$.
\begin{theorem}[Lifting theorem \cite{GLM+15}] \label{thm:simu}
For any $\epsilon>0$ define $b\coloneqq \Theta(\log(n/\epsilon))$ (with a large enough implicit constant). For every partial $f\colon\{0,1\}^n\to\{0,1,*\}$ we have
\[
\R_\epsilon(f\circ \IP_b) \geq \Omega(\deg_\epsilon(f)\cdot b).
\]
\end{theorem}

\begin{proof}[\bf Proof of \autoref{thm:lifted-lb}]
Noting that $\R(f) = \R(\neg f)$ and $\R_{1/n}(f) \le O(\R(f)\cdot \log(n)))$ this follows immediately from the lifting theorem together with the junta degree lower bounds from 
\autoref{thm:junta_deg}.
\end{proof}

\section{Quantum upper bounds} \label{sec:quantum-ub}
In this section we explicitly define the lookup functions we will use for our 
bounded-error quantum and exact quantum vs. randomized communication complexity separations.  We show 
upper bounds on the quantum communication complexity of these functions which, together with the randomized 
lower bounds from \autoref{sec:R_lower}, will finish the proofs of \autoref{thm:q-vs-r} and \autoref{thm:qe-vs-r}.  

First we need some preliminary results about the behavior of quantum query algorithms under composition and 
the relation between the quantum query complexity of a function $f$ and the quantum communication complexity of 
a lifted version of $f$ after composition with a communication function.

\begin{fact}[Composition of quantum query complexity~\cite{Rei11}]
\label{fact:QdtComp}
Let $f : \B^n \rightarrow \{0,1,*\}$ and $g: \B^m \rightarrow \B$.  Then $\Qdt(f \circ g^n) = O(\Qdt(f) \Qdt(g))$ 
and $\QEdt(f \circ g^n) = O(\QEdt(f) \QEdt(g))$.
\end{fact}

\begin{fact}[Composition with query function~\cite{BCW98}]
\label{fact:QComp}
Let $f\colon \{0,1\}^n\to \{0,1,*\}$ be a (partial) function.
For $i\in[n]$, let $F_i\colon \X\times \Y \to \{0,1,*\}$ be a communication problem.
Then $\Q(f\circ (F_1,\dots,F_n)) = O(\Qdt(f)\log \Qdt(f)\cdot \max_i \Q(F_i)\log n)$ and 
$\Q_E(f\circ (F_1,\dots,F_n)) = O(\QEdt(f)\cdot \max_i \Q_E(F_i)\log n)$.
\end{fact}

\subsection{Proof of \autoref{thm:q-vs-r}}
\label{sec:q-vs-r}
Let $F=\STR \circ \IP_b$ as defined in~\autoref{eq:str}, for $b=\Theta(\log n)$.  Let $c=10\log n$.  The definition of the 
family of functions $\G=\{G_0, \ldots, G_{2^c-1}\}$, closely resembles the construction of cheat sheet functions.  The 
most difficult property to achieve is to make $\G$ consistent outside $F$.  We do this by defining $G_i(\x,u,\y,v)$ to be 
$1$ if and only if $u \oplus v$ certifies that each $(x_i,y_i)$ is in the domain of $F$ (all functions $G_i$ will be the same).  
This condition naturally enforces 
consistency outside $F$.  We further require that $u \oplus v$ certifies this in a very specific fashion.  This is done 
so that the players can check $u \oplus v$ has the required properties efficiently using a distributed version of 
Grover's search algorithm.

We first define a helper function which will be like $G_i$ but just works to certify that a single copy $(x_j,y_j)$ of the 
input is in $\dom(F)$.  Let
\[
P \colon \left(\B^{bn^3} \times \B^{n (n\log n +1)}\right) \times \left(\B^{bn^3} \times \B^{n (n\log n +1)}\right)\rightarrow \B.
\]  
This function 
will be defined such that 
$P(x,u,y,v)=1$ if and only if $(x,y) \in \dom(F)$ is witnessed by $u \oplus v$ in a specific fashion, described next.
Decompose $x \in \B^{bn^3}$ as $x=(x_{i,j,k})_{i,j,k \in [n]}$ where each $x_{ijk} \in \B^{b}$, and similarly for $y$.  Let 
$z_{ijk} = \IP_b(x_{ijk},y_{ijk})$ for $i,j,k \in [n]$, and $z_i=\OR_n \circ \AND_n(z_{i11}, \ldots, z_{inn})$ for $i \in [n]$.  
Now $(x,y)$ will be in the domain of $F$ if and only if $(z_1, \ldots, z_n)$ is in the domain of 
$\SIMON_n$.  

If the players know $(z_1, \ldots, z_n)$ then they can easily verify if it is in $\dom(\SIMON_n)$.  
Let $w = u \oplus v$ and decompose this as $w=(q, \C)$, where $q \in \B^n$ and $\C=(C_1, \ldots, C_n)$ 
with each $C_i \in [n]^n$.  Intuitively, $q$ can be thought of as the purported value of $(z_1,\ldots, z_n)$, and 
$C_i$ as a ``certificate'' that $q_i=z_i$.  The function evaluates to $1$ if these certificates check out.

Formally, $P(x,u,y,v) =1$ if and only if 
\begin{enumerate}
\item $q \in \dom(\SIMON_n)$ 
\item for all $i \in [n]$: if $q_i=1$ then $C_i = (j,0,\ldots, 0)$ and 
$z_{ijk}=1$ for all $k \in [n]$, and if $q_i=0$ then $C_i=(t_1, \ldots, t_n)$ and $z_{ijt_j}=0$ for all $j \in [n]$.  
\end{enumerate}
Note that~(2) ensures that if $P(x,u,y,v)$ accepts then $z_i=q_i$ for all $i \in [n]$.

Finally, we can define $G_i$ for $i \in \{0, \ldots, 2^c-1\}$:
$G_i(\x,u_1, \ldots, u_c,\y, v_1, \ldots, v_c)=1$ if and only if $P( (x_j, u_j), (y_j,v_j))=1$ for all $j\in [c]$.  

\begin{claim}
\label{claim:Gcons}
The family of functions $\G$ defined above is consistent outside of $F$ and is a nontrivial XOR function.
\end{claim}
\begin{proof}
Each $G_i$ is an XOR function by definition.  Also, if $F^c(\x,\y) \not \in \B^c$ because (say) $(x_j, y_j) \not \in \dom(F)$, 
then $P( (x_j,u), (y_j,v))$ will always evaluate to $0$ no matter what $u,v$.  This is because $P( (x_j,u), (y_j,v))$ can only 
evaluate to $1$ if $u \oplus v = (q,\C)$ where $\C$ certifies that $z_i=q_i$ for all $i \in [n]$ as in item~(2) above.  If this 
holds, then $P$ will reject when $q=(z_1, \ldots, z_n) \not \in \dom(F)$.  This means $\G$ is consistent outside 
$F$.  

Finally, let $(\x, \y) \in \dom(F^c)$.  Then there will exist $u,v$ such that $u \oplus v$ provides correct 
certificates of this, and $u', v'$ providing incorrect certificates.  Thus each $G_i$ is nontrivial.
\end{proof}

We can now finish the separation.

\begin{theorem}
Let $F=\STR \circ \IP_b$ be defined as in~\eq{str} for $b=\Theta(\log n)$, 
$\G$ be the family of functions defined above, and $F_\G$ be the $(F, \G)$-lookup function.  Then $F_\G$ is a total function satisfying
\[\Q(F_\G)= \tO(bn) = \tO(n) \qquad \textrm{and} \qquad \R(F_\G) = \tOmega(n^{2.5}).\]
\end{theorem}

\begin{proof}
We start with the randomized lower bound.  As $c=10 \log n \ge \R(F)$ we can apply~\autoref{thm:ICCS} to obtain 
$\R(F_\G) =\tOmega(\R(F))= \tOmega(n^{2.5})$ 
by~\autoref{thm:lifted-lb}.

Now we turn to the quantum upper bound.
By \autoref{thm:lookup_upper} it suffices to show $\Q(F) = \widetilde O(bn)$ and $\max_s \Q(G_s) =\widetilde O(bn)$.
As $\Qdt(\SIMON_n) = O(\log n \log \log n)$ and $\Qdt(\OR_n \circ \AND_n)=O(n)$, by the composition theorem 
\autoref{fact:QdtComp} $\Q(\STR) = \widetilde O(n)$.  Thus $\Q(F) = \widetilde O(bn)$ by 
\fct{QComp}, as $\Q(\IP_b) \le b$.

We now turn to show $\max_s \Q(G_s) =\widetilde O(bn)$.  Fix $s$ and
let the input to $G_s$ be $(\x,\fu,\y,\fv)$.  For each $\ell \in [c]$ the players do the following procedure to evaluate 
$P(x_\ell, u_\ell, y_\ell, v_\ell)$.  For ease of 
notation, fix $\ell$ and let $x=x_\ell, y=y_\ell, u=u_\ell, v=v_\ell$.  As above, let $x=(x_{i,j,k})_{i,j,k \in [n]}$ where each 
$x_{ijk} \in \B^{b}$ and similarly for $y$,
$z_{ijk} = \IP_b(x_{ijk},y_{ijk})$ for $i,j,k \in [n]$, and 
$z_i=\OR_n \circ \AND_n(z_{i11}, \ldots, z_{inn})$ for $i \in [n]$.  Also let $w = u \oplus v$ and $w=(q, \C)$ where 
$\C=(C_1, \ldots, C_n)$ and each $C_i \in [n]^n$.  We will further decompose $C_i=(C_{i1}, \ldots, C_{in})$.

Alice and Bob first exchange $n$ bits to learn $q$.  If $q \not \in \dom(\SIMON_n)$ they reject.  Otherwise, they 
proceed to 
check property~(2) above, that $C_i$ certifies that $q_i=z_i$ for all $i \in [n]$.  They view this as a search problem on 
$n^2$ items 
$g_{i,t} \in \B$ for $i,t \in [n]$.  If $q_i=1$ then $g_{i,t}=1$ if and only if $z_{itC_{it}}=1$.  If $q_i=0$ then 
$g_{i,t}=1$ if and only if $z_{itC_{it}}=0$.  Then $(x,u,y,v)$ satisfy property~(2) in the definition of $P$ if and only 
if $g_{i,t}=1$ for all $i,t \in [n]$.  Each $g_{i,t}$ can be evaluated using $O(b + \log n)$ bits of communication.
Hence, using Grover search and 
\fct{QComp}, it takes $\widetilde O(b n)$ qubits of quantum communication to verify that all $g_{i,t}=1$.
\end{proof}

\subsection{Proof of \autoref{thm:qe-vs-r}}
We now turn to the separation between exact quantum and randomized communication complexities.
Fix $n$ and $m \le n$ and let $H = \pTR_{n,m} \circ \IP_b$ 
where $b = \Theta(\log n)$.  The separation will be given for a $(H,\G)$ lookup function with 
$m=\Theta(\sqrt{n})$, where the family 
$\G$ is defined in a similar way as in \autoref{sec:q-vs-r}.

For $c = 10 \log n$ let $\G = \{ G_0, \ldots, G_{2^c-1} \}$, where the functions 
$G_i : (\{0,1\}^{cbnm} \times [m]^{cn}) \times (\{0,1\}^{cbnm} \times [m]^{cn}) \rightarrow \{0,1\}$ will not depend 
on $i$.  To define $G_i$ it is useful to first define a helper function 
$P: (\{0,1\}^{bnm} \times ([m]^{n})) \times (\{0,1\}^{bnm} \times ([m]^{n})) \rightarrow \{0,1\}$, 
where $P(x,u,y,v)=1$ if and only if $u \oplus v$ witnesses that $(x,y) \in \dom(F)$ in a specific fashion, described next.  

Decompose $x = (x_{11},\ldots, x_{nm})$ where each $x_{ij} \in \{0,1\}^b$, and similarly for $y$.  Further let 
$x_i=(x_{i1}, \ldots, x_{im})$ for $i \in [n]$, and similarly for $y_i$.  Let $w= u \oplus v$ 
and decompose $w = (C_1, \ldots, C_n)$ where each $C_i \in [m]$. To show that 
$z=\AND_m \circ \IP_b^m(x_1, y_1), \ldots, \AND_m \circ \IP_b^m(x_n, y_n)$ is in the domain of $\pOR_n$ we need to 
point out $n-1$ zeros of $z$.  Formally, 
$P(x,u,y,z) = 1$ if and only if $u \oplus v = (C_1, \ldots, C_n)$ and $\IP_b(x_{iC_i}, y_{iC_i}=0$ for at least 
$n-1$ many $i's$.  

Finally, we can define $G_i$ for $i \in \{0, \ldots, 2^c-1\}$:
$G_i(\x,u_1, \ldots, u_c,\y, v_1, \ldots, v_c)=1$ if and only if $P( (x_j, u_j), (y_j,v_j))=1$ for all $j\in [c]$.  Note that 
if $(H(x_1, y_1), \ldots, H(x_c,y_c)) \not \in \{0,1\}^c$ then $G_i(\x,\fu,\y,\fv)=0$ for all $i$, thus $\G$ is consistent 
outside of $H$.  Furthermore $\G$ is an XOR function by definition and is nontrivial.

\begin{theorem}
\label{thm:qe_upper}
Fix $n$ and $m \le n$ and let $H = \pTR_{n,m} \circ \IP_b$ where $b = \Theta(\log n)$.
Let $H_\G$ be the $(H,\G)$ lookup function defined above.  Then $H_\G$ is a total function 
satisfying
\[\Q_E(H_\G) = \tO(\sqrt{n}m + n) \qquad \textrm{and} \qquad \R(H_\G) = \tOmega(mn).\]
In particular, $\R(H_\G) = \tOmega(\Q_E(H_\G)^{1.5})$ when $m = \Theta(\sqrt{n})$.
\end{theorem}

\begin{proof}
As $\G$ is a nontrivial XOR function consistent outside of $H$, the lower bound follows from 
\autoref{thm:lifted-lb} and \autoref{thm:ICCS}.

For the upper bound, by \autoref{thm:lookup_upper} it suffices to show upper bounds on 
$\Q_E(H)$ and $\max_{s} \Q_E(G_s)$.  
That $\Q_E(H)=\tO(\sqrt{n}m)$ follows by the composition of exact quantum query complexity and \autoref{fact:QComp},
as $\QEdt(\pOR_n) = \sqrt{n}$ and $\QEdt(\AND_n) = m$.  

For the second part we show that $\Q_E(G_i) \le \D(G_s) = \tO(n)$.  As all functions in the family $\G$ are the same, 
we drop the subscript $i$.  To evaluate $G(\x,\fu,\y,\fv)$ for each $j \in [c]$ the players do the following to evaluate 
$P(x_j, u_j, y_j, v_j)$.  They exchange $u_j, v_j$ with $O(n\log(m))$ to learn $u_j \oplus v_j = (C_1, \ldots, C_n)$.  
For each $t \in [n]$ they evaluate $\IP_b(x_{tC_t}, y_{tC_t})$.  If at least $n-1$ of these values are zero, then they accept.
\end{proof}

\section{Partitions vs.\ randomized communication}
\label{sec:RvsUN}

In this section, we prove \autoref{thm:un-vs-r}, which we restate for convenience:
\unvsr*

The proof closely follows the analogous result obtained for query complexity in~\cite{AKK16} using the cheat sheet 
framework.  For a total communication function $F$, we will define a special case of $(F, \G)$-lookup functions that are a 
communication analog of cheat sheets in query complexity.  
\begin{definition}[Cheat sheets for total functions]
	\label{def:cheat}
	Let $F\colon\X \times \Y \to \B$ be a total function.  Fix a cover $\mathcal{R}=\{R_0, \ldots, R_{2^{N(F)}-1}\}$ 
	of $\X \times \Y$ by rectangles monochromatic for $F$.  Let $N=\min\{\log |\X|, \log |\Y|\}$ and $c=10\log N$.  
	Define a function 
	\[
	G: (\X^c \times \B^{cN(F)}) \times (\Y^c \times \B^{cN(F)}) \rightarrow \B
	\]
	where
     	$G(x_1,\ldots, x_c, a_1, \ldots, a_c, y_1, \ldots, y_c, b_1, \ldots, b_c) = 1$ if and only if 
	$(x_i, y_i) \in R_{a_i \oplus b_i}$ 
	for all $i=1,\ldots, c$.  The \emph{cheat sheet} function $F_\CS$ of $F$ is the $(F, \{G_0, \ldots, G_{2^c-1}\})$ 
	lookup function where $G_i=G$ for all $i$.  In other words, 
	$F_\CS(x_1, \ldots, x_c, u_0, \ldots, u_{2^c-1}, y_1, \ldots, y_c, v_0, \ldots, v_{2^c-1})$ evaluates to 
	$G(x_1, \ldots, x_c, u_\ell, y_1,\ldots, y_c, v_\ell)$, where $\ell = (F(x_1, y_1), \ldots, F(x_c,y_c))$.
\end{definition}

\begin{remark}
Note that $F_\CS$ is in particular a $(F,\G)$-lookup function where $\G$ is a nontrivial $\XOR$ family (\autoref{def:xor}), 
thus \autoref{thm:ICCS} applies.  Further letting $\X' \times \Y'$ be the domain of $F_\CS$, note that 
$N'=\min \{ \log |\X'|, \log |\Y'|\} = O (c N + c \cdot 2^c N(F)) = O(N^{12})$.
\end{remark}

Recall that the function $\TR_{n^2}$ on $n^2$ input bits is the composition $\OR_n \circ \AND_n$.
The separating function of \autoref{thm:un-vs-r} is constructed by starting with disjointness on $n$ variables
and alternately taking the cheat sheet function of it and composing $\TR_{n^2}$ with it. Repeating this process 
gives a function with a larger and larger gap between $\R$ and $\UN$, converging to a quadratic gap between 
these measures.

To prove this result, we first need to understand how the composition
operations affect $\R$ and $\UN$. We start with $\UN$, for which we wish
to prove an upper bound.

\begin{lemma}[AND/OR composition]
\label{lem:composition}
For any communication function $F$, the following bounds hold:

\vspace{0.5em}
\centering
\begin{minipage}[t]{0.5\textwidth}
\begin{itemize}[noitemsep]
\item $\N_0(\AND_n \circ F) \leq \N_0(F) + \log n$
\item $\N_1(\AND_n \circ F) \leq n \N_1(F)$
\item $\UN_0(\AND_n \circ F) \leq \UN_0(F)+(n-1)\UN_1(F)$
\item $\UN_1(\AND_n \circ F) \leq n \UN_1(F)$
\end{itemize}
\end{minipage}%
\begin{minipage}[t]{0.5\textwidth}
\begin{itemize}[noitemsep]
\item $\N_0(\OR_n\circ F) \leq n \N_0(F)$
\item $\N_1(\OR_n \circ F) \leq \N_1(F) + \log n$
\item $\UN_0(\OR_n \circ F) \leq n\UN_0(F)$
\item $\UN_1(\OR_n \circ F) \leq (n-1)\UN_0(F)+\UN_1(F)$
\end{itemize}
\end{minipage}
\end{lemma}
\begin{proof}
We prove the statements for the functions of the form $\AND_n \circ F$. The proofs for the functions $\OR_n \circ F$
are immediate by duality.  A 0-certificate for $\AND_n \circ F$ on input $((x_1, y_1),\ldots, (x_n,y_n))$
can be the index $i$ such that $F(x_i,y_i)=0$, and 0-certificate for $(x_i,y_i)$ on $F$.  A 1-certificate for $\AND_n \circ F$ 
can be 1-certificates for each $(x_i,y_i)$ on $F$, for $i=1,\ldots, n$. For an unambiguous 0-certificate we can choose an 
unambiguous 0-certificate for $(x_i,y_i)$ on $F$ for the least $i$ such that $F(x_i,y_i)=0$, and unambiguous 1-certificates
for $(x_j,y_j)$ on $F$ for all $j=1,\ldots, i-1$.  For an unambiguous 1-certificate we can choose 
an unambiguous 1-certificate for each $(x_i,y_i)$ on $F$, for $i=1,\ldots, n$.
\end{proof}
We have the following corollary.

\begin{corollary}[Tribes composition]
\label{cor:tribes-composition}
Let $\TR_{n^2} = \OR_n \circ \AND_n$. For any function $F$, we have:
\begin{itemize}[noitemsep,topsep=5pt]
\item $\N(\TR_{n^2} \circ F) = O(n \N(F)+ n \log n)$
\item $\UN(\TR_{n^2} \circ F) \leq n \UN_0(F) + n^2 \UN_1(F)$
\end{itemize}

\end{corollary}

We now analyze the properties of $\N$ and $\UN$ under the cheat sheet operation.

\begin{lemma}[Nondeterministic complexity of cheat sheet functions]
\label{lem:cheat}
Let $F_\CS$ be the cheat-sheet version of a total function $F\colon \X \times \Y 
\rightarrow \{0,1\}$ where $N = \min\{\log|\X|,\log|\Y|\}$. Then
\[
\N(F_\CS) = O(\N(F) \log N), \qquad \UN_1(F_\CS) = O(\N(F) \log N), \qquad \UN_0(F_\CS) = O(\UN(F) \log N).
\]
\end{lemma}

\begin{proof}
We first upper bound $\N_1(F_\CS)$ by giving a protocol.  Let $\x = (x_1, \ldots, x_c), \y=(y_1, \ldots, y_c)$ and 
consider an input $(\x, u_0, \ldots, u_{2^c-1}, \y, v_0, \ldots, v_{2^c-1})$ to $F_\CS$.   The prover provides a proof 
of the 
form $(\ell, a, b)$ where $\ell \in \{0,\cdots, 2^c-1\}, a,b \in \B^{cN(F)}$.  Note that the length of the proof is 
$O(c N(F)) = O(N(F) \log N)$.  The players accept if and only if $u_\ell =a, v_\ell=b$, and 
$a \oplus b$ provides certificates that $F(x_i,y_i) = \ell_i$ for all $i=1, \ldots, c$.  If $F_\CS$ evaluates to $1$ on this 
input, a valid proof always exists by giving $\ell = F^c(\x,\y)$ and $a=u_\ell, b=v_\ell$.  On the 
other hand if $F_\CS$ evaluates to $0$ on this input, then by definition of the cheat sheet function for any 
message $(\ell, a,b)$ it cannot be that $a,b$ agree with $u_\ell, v_\ell$ and 
that $a \oplus b$ certifies that $F^c(\x, \y) = \ell$.  

This protocol is in fact unambiguous.  Say that $F_\CS$ evaluates to $1$ on the input 
$(\x, \fu, \y, \fv)$ and let $\ell=F^c(\x,\y)$.  A valid proof is given by 
$(\ell, u_\ell, v_\ell)$.  Consider another proof $(\ell', a, b)$.  First, if $\ell' \ne \ell$, then $a \oplus b$ cannot certify that 
$F^c(\x,\y) = \ell'$, as $F^c(\x,\y)=\ell$.  Now if $\ell' = \ell$, then the players will only accept if $a=u_\ell$ and 
$b=v_\ell$.  Thus there is a unique accepting proof.

We now turn to bound the $\N_0$ complexity.  Fix a cover $C_1, \ldots, C_{2^{N(F)}}$ of $F$ by 
monochromatic rectangles.  In this case the prover provides a message of the form $(\ell, i_1, \ldots, i_c,a,b)$, 
where $\ell \in \{0,\ldots, 2^c-1\}, i_j \in \B^{N(F)}, a,b \in \B^{cN(F)}$.  Thus the length of the proof is 
$O(cN(F))= O(N\log N)$.  Alice and Bob accept if and only if 
\begin{enumerate}
\item $(x_j, y_j) \in C_{i_j}$ for all $j=1, \ldots, c$.
\item $C_{i_j}$ is $\ell_j$-monochromatic on $F$ for $j=1, \ldots, c$, 
\item $u_\ell =a, v_\ell=b$ and $a \oplus b$ does not provide valid certificates that $F^c(\x,\y)=\ell$.  
\end{enumerate}
If $F_\CS(\x,\fu,\y,\fv)=0$ then there is a valid proof by giving 
$\ell =F^c(\x,\y)$, providing valid certificates for these values, and giving $u_\ell, v_\ell$.  On the other hand, if 
$F_\CS(\x,\fu, \y, \fv)=1$, then if the steps~1,2 of the verification pass then it must be the case that $a,b$ do 
not agree with $u_\ell, v_\ell$, as in this case $u_\ell \oplus v_\ell$ do provide valid certificates.

To upper bound the $\UN_0$ complexity, the protocol is exactly the same except now a partition $R_1, \ldots, 
R_{\chi(F)}$ of rectangles monochromatic for $F$ is used instead of a cover.  In this case, there is a unique 
choice of witnesses $(i_1, \ldots, i_c)$ to certify the correct value $F^c(\x,\y)=\ell$.  The second part $(a,b)$ of a valid 
proof is also uniquely specified as it must agree with the part of the input $(u_\ell, v_\ell)$.
\end{proof}

\begin{corollary}\label{cor:UN}
	For any total function $F\colon \X \times \Y \rightarrow \{0,1\}$
	with $N = \min\{\log|\X|,\log|\Y|\}$, we have
	\begin{itemize}[noitemsep,topsep=5pt]
		\item $\UN(\TR_{n^2}\circ F_\CS)=O(n\UN(F)\log N+n^2\N(F)\log N)$
		\item $\N(\TR_{n^2}\circ F_\CS)=O(n\N(F)\log N)$.
	\end{itemize}
\end{corollary}

We put these together to get an upper bound on $\UN$ for the iterated function.
Let $F_0=\Disj_n$ and $F_{i+1}\defeq\TR_{n^2}\circ (F_i)_\CS$
for all $i\geq 0$. The function $F_k$ for appropriately chosen $k$ will
provide the near-quadratic separation.

\begin{claim}
\label{claim:un}
	There is a constant $a$ such that
	for any $k\geq 0$, we have
		\begin{itemize}[noitemsep,topsep=5pt]
			\item $\UN(F_k)=O(n^{k+2}a^k k^k\log^k n)$
			\item $\N(F_k)=O(n^{k+1}a^k k^k \log^k n).$
		\end{itemize}
	When $k$ is constant, these simplify to $\tO(n^{k+2})$ and
	$\tO(n^{k+1})$, respectively.
\end{claim}

\begin{proof}
	This follows from \autoref{cor:UN} by induction on $k$.
	In the base case, we have $\N(\Disj_n)=O(\UN(\Disj_n))=O(n)$.
	The induction step follows immediately from \autoref{cor:UN}.
	The only subtlety is the size of $N$, which increases polynomially with
	each iteration, which means $\log N=O(k\log n)$. This gives the
	$a^kk^k\log^k n$ factor.
\end{proof}

Next, we prove a lower bound on $\R(F_k)$. To do this, we need to get a handle
on the behavior of $\R$ when the function is composed with $\AND_n$ and
$\OR_n$. We use the following definition and fact.

\begin{definition}
	Let $F\colon \X \times \Y \rightarrow \{0,1,*\}$
	be a (partial) function and let $\eps\in(0,1/2)$.
	For any protocol $\Pi$ and any $b\in\{0,1\}$,
	\[\IC^b(\Pi)\defeq\max_{\mu\text{ on } F^{-1}(b)}
	\IC^\mu(\Pi).\]
\end{definition}
The following claim shows a composition result for one-sided information complexity. A result similar in spirit for the $\OR_n \circ \AND$ function was shown by~\cite{BJKS04}.
\begin{claim}[Composition] \label{fact:AND}
	Let $F\colon \X \times \Y \rightarrow \{0,1,*\}$
	be a (partial) function, and let $\eps\in(0,1/2)$ be a constant.
	For any protocol $\Pi$ for $\OR_n\circ F$ with
	worst case error at most $\eps$, there is a protocol $\Pi^\prime$
	for $F$ with worst error at most $\eps$ such that
	\[\IC^0(\Pi^\prime)=O(\IC^0(\Pi)/n)\quad \textrm{and} \quad
	\CC(\Pi^\prime)=O(\CC(\Pi)).\]
	Similarly, if $\Pi$ is a protocol for $\AND_n\circ F$ with
	worst case error at most $\eps$,
	there is a protocol $\Pi^\prime$ for $F$ with
	worst case error at most $\eps$, such that
	\[\IC^1(\Pi^\prime)=O(\IC^1(\Pi)/n)\quad \textrm{and} \quad
	\CC(\Pi^\prime)=O(\CC(\Pi)).\]
\end{claim}
\begin{proof}
We show the result for $\OR_n\circ F$. The result for $\AND_n\circ F$ follows similarly.  Let $\mu$ be a distribution on $F^{-1}(0)$. We will exhibit a protocol $\Pi^1$ for $F$ with worst case error at most $\eps$, such that 
\[\IC^\mu(\Pi^1)=O(\IC^0(\Pi)/n)\quad \textrm{and} \quad
	\CC(\Pi^1)=O(\CC(\Pi)).\]
	The desired result then follows from~\fullbref{fact:equiv} and~\fullbref{fact:boost}.
Let us define random variables:
\begin{enumerate}
\item  $XY=(X_1Y_1, \ldots, X_nY_n)$ where each $(X_iY_i) \sim \mu$ and i.i.d. 
\item $D=(D_1, \ldots, D_n)$ where each $D_i$ is uniformly distributed in $\{A,B\}$ and i.i.d. 
\item $U=(U_1, \ldots, U_n)$ where for each $i$, $U_i = X_i$ if $D_i=A$ and $U_i=Y_i$ if $D_i=B$. 
\end{enumerate}
Using~\fullbref{fact:infoind} we have,
$$ I(XY:\Pi~|~DU) \geq \sum_{i=1}^n   I(X_iY_i:\Pi~|~DU).$$
This implies there exists $j \in [n]$ such that  
\begin{align*}
\frac{1}{n} I(XY:\Pi~|~DU) &\geq I(X_jY_j:\Pi~|~DU)  \\
&=  I(X_jY_j:\Pi~|~DjUj D_{-j}U_{-j}) \\
& =   \frac{1}{2} \left(I(X_j:\Pi~|~Y_j D_{-j}U_{-j})  + I(X_j:\Pi~|~Y_j D_{-j}U_{-j}) \right) \\
& =   \frac{1}{2} \left(I(X_j:\Pi D_{-j}U_{-j} ~|~Y_j )  + I(X_j:\Pi D_{-j}U_{-j}~|~Y_j) \right). & \hspace{-4.3pt}\mbox{(\fullref{fact:barhopping})}
\end{align*}
Define protocol $\Pi^1$ as follows. Alice and Bob insert their inputs at the $j$-th coordinate and generate $(D_{-j}U_{-j})$ using public-coins. They go ahead and simulate $\Pi$ afterwards.
From above we have 
\begin{equation}
\frac{1}{n} I(XY:\Pi~|~DU) \geq \frac{1}{2} \IC^\mu(\Pi^1). \label{eq:bjks}
\end{equation}
It is clear that $\CC(\Pi^1) \leq \CC(\Pi)$ and the worst case error of $\Pi^1$ is upper bounded by the worst case error of $\Pi$.
Consider,
\begin{align*}
\IC^0(\Pi) & \geq \IC^{XY}(\Pi) \\
& = I(X:\Pi~|~Y) + I(Y:\Pi~|~X) \\
& =   I(X:\Pi~|~Y) + I(DU : \Pi~|~ XY) + I(Y:\Pi~|~X) + I(DU : \Pi~|~ XY) \hspace{-3.5em} & \quad (DU \leftrightarrow XY \leftrightarrow \Pi)  \\
& = I(XDU:\Pi~|~Y) +  I(YDU:\Pi~|~X)  & \quad \mbox{(\fullref{fact:chain-rule})}\\
& \geq I(X:\Pi~|~YDU) +  I(Y:\Pi~|~XDU)  & \quad \mbox{(\fullref{fact:barhopping})}\\
& = I(X:\Pi~|~YDU) +  I(X:Y~|~DU) +  I(Y:\Pi~|~XDU)  & \quad (X \leftrightarrow DU \leftrightarrow Y) \\
& = I(X:\Pi Y~|~DU) +  I(Y:\Pi~|~XDU) & \quad \mbox{(\fullref{fact:chain-rule})}\\
& \geq I(X:\Pi~|~DU) +  I(Y:\Pi~|~XDU) & \quad \mbox{(\fullref{fact:mono})}\\
& = I(XY:\Pi~|~DU). & \quad \mbox{(\fullref{fact:chain-rule})}
\end{align*}
This along with Eq.~\eqref{eq:bjks} shows the desired.
\end{proof}

To be able to use this, we need a way of converting between
$\IC^0$, $\IC^1$, and $\IC$.
The following fact was shown by \cite[Corollary 18]{GJPW15} using the
``information odometer'' of Braverman and
Weinstein \cite{BW15} (the upper bound on
$\CC(\Pi^\prime)$ was not stated explicitly
in \cite{GJPW15}, but it traces back to \cite[Theorem 3]{BW15}, which was used in \cite{GJPW15}).

\begin{fact}\label{fact:IC_onesided}
	Let $F\colon \X \times \Y \rightarrow \{0,1\}$
	be a function.
	Let $1/2> \delta > \eps >0$ and $b \in \{0,1\}$. Then
	for any protocol $\Pi$ with $\err(\Pi)<\eps$,
	there is a protocol $\Pi^\prime$ with
	$\err(\Pi^\prime)<\delta$ such that
	\[\IC(\Pi^\prime) = O(\IC^b(\Pi)+\log\CC(\Pi))\quad \textrm{and} \quad
	\CC(\Pi^\prime) = O(\CC(\Pi)\log\CC(\Pi)).\]
\end{fact}


\begin{theorem}
\label{thm:recursion}
There is a constant $b$ such that
for every $k \leq n^{1/10}$, we have
\[\R(F_k) = \Omega\left(\frac{n^{2k+1}}{b^k k^{3k}\log^{3k} n}\right).\]
\end{theorem}
\begin{proof}
	Consider the protocol $\Pi$ for $F_k$
	with error at most $1/3$ such that $\CC(\Pi)=\R(F_k)$,
	and hence $\IC(\Pi)=O(\R(F_k))$. Recall that
	$F_k=\TR_{n^2}\circ (F_{k-1})_\CS=\OR_n\circ\AND_n\circ(F_{k-1})_\CS$.
	Using \autoref{fact:AND}, we get a protocol $\Pi'$ for
	$\AND_n\circ(F_{k-1})_\CS$ with $\err(\Pi')\leq 1/3$,
	$\IC^0(\Pi')=O(\IC^0(\Pi)/n)=O(\IC(\Pi)/n)=O(\R(F_k)/n)$,
	and $\CC(\Pi')=O(\CC(\Pi))$. Using \autoref{fact:IC_onesided},
	we get a protocol $\Pi''$ for $\AND_n\circ(F_{k-1})_\CS$ with
	$\err(\Pi'')\leq 2/5$,
	$\IC(\Pi'')=O(\IC^0(\Pi')+\log\CC(\Pi'))=O(\R(F_k)/n+\log\R(F_k))$,
	and $\CC(\Pi'')=O(\CC(\Pi')\log\CC(\Pi'))=O(\R(F_k)\log\R(F_k))$.
	Using \fullbref{fact:boost}, we get a protocol $\Pi'''$ for
	$\AND_n\circ(F_{k-1})_\CS$ with $\err(\Pi''')\leq 1/3$,
	$\IC(\Pi''')=O(\R(F_k)/n+\log\R(F_k))$, and $\CC(\Pi''')=O(\R(F_k)\log\R(F_k))$.
	
	We can repeat this process to strip away the $\AND_n$; that is, we use
	\autoref{fact:AND}, \autoref{fact:IC_onesided}, and \fullbref{fact:boost}
	to get a protocol $\Pi''''$ for $(F_{k-1})_\CS$ with $\err(\Pi'''')\leq 1/3$,
	$\IC(\Pi'''')=O(\R(F_k)/n^2+\log\R(F_k))$, and
	$\CC(\Pi'''')=O(\R(F_k)\log^2\R(F_k))$. Then \autoref{thm:ICCS}
	gives a protocol $\Pi'''''$ for $F_{k-1}$ with $\err(\Pi''''')\leq 1/3$,
	$\IC(\Pi''''')=O((\R(F_k)\log^3 N)/n^2+\log\R(F_k)\cdot \log^3 N)$,
	and $\CC(\Pi''''')=O(\R(F_k)\log^2 \R(F_k)\log^2 N)$,
	where $N$ is the input size of $F_{k-1}$. Here $N=n^{O(k)}$, so
	$\log N=O(k\log n)$ and $\log\R(F_k)=O(k\log n)$, and hence
	$\IC(\Pi''''')=O((\R(F_k)k^3\log^3 n)/n^2+k^4\log^4 n)$ and
	$\CC(\Pi''''')=O(\R(F_k)k^4\log^4 n)$.
	
	We now repeat this $k$ times to get a protocol $\Psi$ for $F_0=\Disj_n$.
	Then we have
	$\CC(\Psi)=O(b^k\R(F_k)k^{4k}\log^{4k} n)$ for some constant $b$,
	and the communication
	complexity of every intermediate protocol in the construction is also
	at most $O(b^k\R(F_k)k^{4k}\log^{4k} n)$. To calculate $\IC(\Psi)$,
	note that each iteration divides $\IC$ by $n^2$, adds a $\log \CC$ term,
	and multiplies by $k^3\log^3 n$.
	Thus we get, for some constant $b$,
	\[\IC(\Psi)=O\left((\R(F_k) b^k k^{3k}\log^{3k}n)/n^{2k}
	+k^3\log^3 n\cdot\log\CC(\Psi)
	\sum_{i=0}^{k-1}\left(\frac{k^3\log^3 n}{n^2}\right)^i\right).\]
	Since $k=O(n^{1/10})$, the sum is $O(1)$, so we get
	\begin{align*}
		\IC(\Psi)&=O((\R(F_k) b^k k^{3k}\log^{3k}n)/n^{2k}
		+k^3\log^3 n\cdot\log\CC(\Psi))\\
		&=O((\R(F_k) b^k k^{3k}\log^{3k}n)/n^{2k}+
		k^3\log^3 n\cdot (\log\R(F_k)+k\log k+k\log\log n))\hspace{-4em} \\
		&=O((\R(F_k) b^k k^{3k}\log^{3k}n)/n^{2k}+
		k^3\log^3 n\cdot \log\R(F_k)+n^{8/10}).& \text{(since $k=O(n^{1/10})$)}
	\end{align*}
	Now, since $\IC(\Disj_n)=\Omega(n)$, we get either
	$\R(F_k)=2^{\Omega(n/k^3\log^3 n)}=\Omega(2^{\sqrt{n}})$ or
	\[\R(F_k)=\Omega\left(\frac{n^{2k+1}}{b^k k^{3k}\log^{3k} n}\right).\]
	Because $k=O(n^{1/10})$, the value of $2^{\sqrt{n}}$
	is even larger than the desired lower bound, so the desired result follows.
\end{proof}

Finally, we get prove the near-quadratic separation.

\unvsr*

\begin{proof}
	We take $F=F_k$ with $k$ some slowly growing function of $n$.
	In particular, let $k=\sqrt{\frac{\log n}{\log\log n}}$.
	This gives $\R(F_k)\geq \frac{n^{2k+1}}{2^{O(\sqrt{\log n\log\log n})}}$
	and $\UN(F_k)\leq n^{k+2}2^{O(\sqrt{\log n\log\log n})}$,
	so $\log\UN(F_k)=\log^{3/2} n/\log\log^{1/2} n+O(\sqrt{\log n\log\log n})$
	and
	\begin{align*}
	\log\R(F_k)&=2\log^{3/2}n/\log\log^{1/2} n-O(\sqrt{\log n\log\log n})\\
	&=2\log\UN(F_k)-O(\log^{2/3}\UN(F_k)\log\log^{4/3}\UN(F_k)).
	\end{align*}
	Thus
	\[\R(F_k)\geq\UN(F_k)^{2-O(\alpha(\UN(F_k)))}\]
	where $\alpha(x)= \frac{\log\log^{4/3}x}{\log^{1/3} x}=o(1)$.
\end{proof}

\section*{Acknowledgements}
Part of this work was performed when the authors met during the workshop 
``Semidefinite and Matrix Methods for Optimization and Communication" hosted at the Institute for Mathematical Sciences, Singapore. We thank them for their hospitality. R.J would like to thank Ankit Garg for helpful discussions.

This work is partially supported by ARO grant number W911NF-12-1-0486, by the Singapore Ministry of Education and the 
National Research Foundation, also through NRF RF Award No. NRF-NRFF2013-13, and the Tier 3 Grant ``Random 
numbers from quantum processes'' MOE2012-T3-1-009.
This research is also partially supported by the European Commission IST STREP project Quantum Algorithms
(QALGO) 600700 and by the French ANR Blanc program under contract ANR-12-BS02-
005 (RDAM project).
M.G.\ is partially supported by the Simons Award for Graduate Students in TCS.


\DeclareUrlCommand{\Doi}{\urlstyle{sf}}
\renewcommand{\path}[1]{\small\Doi{#1}}
\renewcommand{\url}[1]{\href{#1}{\small\Doi{#1}}}
\newcommand{\eprint}[1]{\href{http://arxiv.org/abs/#1}{\small\Doi{#1}}}
\bibliographystyle{alphaurl}
\phantomsection\addcontentsline{toc}{section}{References} 
\bibliography{cheat}

\newcommand{\etalchar}[1]{$^{#1}$}
\begin{thebibliography}{GLM{\etalchar{+}}15}

\bibitem[AA03]{AA03}
Scott Aaronson and Andris Ambainis.
\newblock Quantum search of spatial regions.
\newblock In {\em Proceedings of the 44th Symposium on Foundations of Computer
  Science (FOCS)}, pages 200--209, 2003.
\newblock \href {http://dx.doi.org/10.1109/SFCS.2003.1238194}
  {\path{doi:10.1109/SFCS.2003.1238194}}.

\bibitem[AA15]{AA15}
Scott Aaronson and Andris Ambainis.
\newblock Forrelation: A problem that optimally separates quantum from
  classical computing.
\newblock In {\em Proceedings of the 47th Symposium on Theory of Computing
  (STOC)}, pages 307--316, 2015.
\newblock \href {http://dx.doi.org/10.1145/2746539.2746547}
  {\path{doi:10.1145/2746539.2746547}}.

\bibitem[ABK16]{ABK16}
Scott Aaronson, Shalev {Ben-David}, and Robin Kothari.
\newblock Separations in query complexity using cheat sheets.
\newblock In {\em Proceedings of the 48th Symposium on Theory of Computing
  (STOC)}, pages 863--876, 2016.
\newblock \href {http://arxiv.org/abs/1511.01937} {\path{arXiv:1511.01937}},
  \href {http://dx.doi.org/10.1145/2897518.2897644}
  {\path{doi:10.1145/2897518.2897644}}.

\bibitem[AKK16]{AKK16}
Andris Ambainis, Martins Kokainis, and Robin Kothari.
\newblock Nearly optimal separations between communication (or query)
  complexity and partitions.
\newblock In {\em Proceedings of the 31st Conference on Computational
  Complexity (CCC)}, volume~50, pages 4:1--4:14, 2016.
\newblock \href {http://dx.doi.org/10.4230/LIPIcs.CCC.2016.4}
  {\path{doi:10.4230/LIPIcs.CCC.2016.4}}.

\bibitem[Amb13]{Amb13}
Andris Ambainis.
\newblock Superlinear advantage for exact quantum algorithms.
\newblock In {\em Proceedings of the 45th Symposium on Theory of Computing
  (STOC)}, pages 891--200, 2013.

\bibitem[AUY83]{AUY83}
Alfred Aho, Jeffrey Ullman, and Mihalis Yannakakis.
\newblock On notions of information transfer in {VLSI} circuits.
\newblock In {\em Proceedings of the 15th Symposium on Theory of Computing
  (STOC)}, pages 133--139, 1983.
\newblock \href {http://dx.doi.org/10.1145/800061.808742}
  {\path{doi:10.1145/800061.808742}}.

\bibitem[BCW98]{BCW98}
Harry Buhrman, Richard Cleve, and Avi Wigderson.
\newblock Quantum vs. classical communication and computation.
\newblock In {\em Proceedings of the 30th Symposium on Theory of Computing
  (STOC)}, pages 63--68, 1998.
\newblock \href {http://dx.doi.org/10.1145/276698.276713}
  {\path{doi:10.1145/276698.276713}}.

\bibitem[BdW02]{BdW02}
Harry Buhrman and Ronald de~Wolf.
\newblock Complexity measures and decision tree complexity: a survey.
\newblock {\em Theoretical Computer Science}, 288(1):21--43, 2002.
\newblock \href {http://dx.doi.org/10.1016/S0304-3975(01)00144-X}
  {\path{doi:10.1016/S0304-3975(01)00144-X}}.

\bibitem[BFNR08]{BFNR08}
Harry Buhrman, Lance Fortnow, Ilan Newman, and Hein R{\"{o}}hrig.
\newblock Quantum property testing.
\newblock {\em SIAM Journal on Computing}, 37(5):1387--1400, 2008.
\newblock \href {http://dx.doi.org/10.1137/S0097539704442416}
  {\path{doi:10.1137/S0097539704442416}}.

\bibitem[BJKS04]{BJKS04}
Ziv {Bar-Yossef}, T.S. Jayram, Ravi Kumar, and D.~Sivakumar.
\newblock An information statistics approach to data stream and communication
  complexity.
\newblock {\em Journal of Computer and System Sciences}, 68(4):702--732, 2004.
\newblock \href {http://dx.doi.org/10.1016/j.jcss.2003.11.006}
  {\path{doi:10.1016/j.jcss.2003.11.006}}.

\bibitem[Bra12]{Bra12}
Mark Braverman.
\newblock Interactive information complexity.
\newblock In {\em Proceedings of the 44th Symposium on Theory of Computing
  (STOC)}, pages 505--524, 2012.
\newblock \href {http://dx.doi.org/10.1145/2213977.2214025}
  {\path{doi:10.1145/2213977.2214025}}.

\bibitem[BW15]{BW15}
Mark Braverman and Omri Weinstein.
\newblock An interactive information odometer and applications.
\newblock In {\em Proceedings of the 47th Symposium on Theory of Computing
  (STOC)}, pages 341--350, 2015.
\newblock \href {http://dx.doi.org/10.1145/2746539.2746548}
  {\path{doi:10.1145/2746539.2746548}}.

\bibitem[CT06]{CT06}
Thomas Cover and Joy Thomas.
\newblock {\em Elements of Information Theory (Wiley Series in
  Telecommunications and Signal Processing)}.
\newblock Wiley-Interscience, 2006.

\bibitem[Das11]{ADasG11}
Anirban DasGupta.
\newblock {\em Probability for Statistics and Machine Learning: Fundamentals
  and Advanced Topics}.
\newblock Springer Texts in Statistics. Springer, 2011.
\newblock \href {http://dx.doi.org/10.1007/978-1-4419-9634-3}
  {\path{doi:10.1007/978-1-4419-9634-3}}.

\bibitem[GJ16]{GJ16}
Mika G{\"o}{\"o}s and T.~S. Jayram.
\newblock A composition theorem for conical juntas.
\newblock In {\em Proceedings of the 31st Conference on Computational
  Complexity (CCC)}, volume~50, pages 5:1--5:16, 2016.
\newblock \href {http://dx.doi.org/10.4230/LIPIcs.CCC.2016.5}
  {\path{doi:10.4230/LIPIcs.CCC.2016.5}}.

\bibitem[GJPW15]{GJPW15}
Mika G{\"o}{\"o}s, T.S. Jayram, Toniann Pitassi, and Thomas Watson.
\newblock Randomized communication vs. partition number.
\newblock {\em Electronic Colloquium on Computational Complexity (ECCC)
  \href{http://eccc.hpi-web.de/report/2015/169/}{TR15-169}}, 2015.

\bibitem[GLM{\etalchar{+}}15]{GLM+15}
Mika G\"{o}\"{o}s, Shachar Lovett, Raghu Meka, Thomas Watson, and David
  Zuckerman.
\newblock Rectangles are nonnegative juntas.
\newblock In {\em Proceedings of the 47th Symposium on Theory of Computing
  (STOC)}, pages 257--266, 2015.
\newblock \href {http://dx.doi.org/10.1145/2746539.2746596}
  {\path{doi:10.1145/2746539.2746596}}.

\bibitem[GPW15]{GPW15}
Mika G{\"o\"o}s, Toniann Pitassi, and Thomas Watson.
\newblock Deterministic communication vs. partition number.
\newblock In {\em Proceedings of the 56th Symposium on Foundations of Computer
  Science (FOCS)}, pages 1077--1088, 2015.
\newblock \href {http://dx.doi.org/10.1109/FOCS.2015.70}
  {\path{doi:10.1109/FOCS.2015.70}}.

\bibitem[JK10]{JK10}
Rahul Jain and Hartmut Klauck.
\newblock The partition bound for classical communication complexity and query
  complexity.
\newblock In {\em Proceedings of the 25th Conference on Computational
  Complexity (CCC)}, pages 247--258, 2010.
\newblock \href {http://dx.doi.org/10.1109/CCC.2010.31}
  {\path{doi:10.1109/CCC.2010.31}}.

\bibitem[JKS03]{JKS03}
T.S. Jayram, Ravi Kumar, and D.~Sivakumar.
\newblock Two applications of information complexity.
\newblock In {\em Proceedings of the 35th Symposium on Theory of Computing
  (STOC)}, pages 673--682, 2003.
\newblock \href {http://dx.doi.org/10.1145/780542.780640}
  {\path{doi:10.1145/780542.780640}}.

\bibitem[Kla10]{Klauck10}
Hartmut Klauck.
\newblock A strong direct product theorem for disjointness.
\newblock In {\em Proceedings of the 42nd Symposium on Theory of Computing
  (STOC)}, pages 77--86, 2010.
\newblock \href {http://dx.doi.org/10.1145/1806689.1806702}
  {\path{doi:10.1145/1806689.1806702}}.

\bibitem[KLO99]{KLO99}
Eyal Kushilevitz, Nathan Linial, and Rafail Ostrovsky.
\newblock The linear-array conjecture in communication complexity is false.
\newblock {\em Combinatorica}, 19(2):241--254, 1999.
\newblock \href {http://dx.doi.org/10.1007/s004930050054}
  {\path{doi:10.1007/s004930050054}}.

\bibitem[KN06]{KN06}
Eyal Kushilevitz and Noam Nisan.
\newblock {\em Communication Complexity}.
\newblock Cambridge University Press, 2006.
\newblock URL: \url{http://books.google.ca/books?id=dHH7rdhKwzsC}.

\bibitem[KS92]{KS92}
Bala Kalyanasundaram and Georg Schintger.
\newblock The probabilistic communication complexity of set intersection.
\newblock {\em SIAM Journal on Discrete Mathematics}, 5(4):545--557, 1992.
\newblock \href {http://dx.doi.org/10.1137/0405044}
  {\path{doi:10.1137/0405044}}.

\bibitem[Raz92]{Raz92}
Alexander Razborov.
\newblock On the distributional complexity of disjointness.
\newblock {\em Theoretical Computer Science}, 106(2):385--390, 1992.
\newblock \href {http://dx.doi.org/10.1016/0304-3975(92)90260-M}
  {\path{doi:10.1016/0304-3975(92)90260-M}}.

\bibitem[Raz99]{Raz99}
Ran Raz.
\newblock Exponential separation of quantum and classical communication
  complexity.
\newblock In {\em Proceedings of the 31st Symposium on Theory of Computing
  (STOC)}, pages 358--367, 1999.
\newblock \href {http://dx.doi.org/10.1145/301250.301343}
  {\path{doi:10.1145/301250.301343}}.

\bibitem[Rei11]{Rei11}
Ben Reichardt.
\newblock Reflections for quantum query algorithms.
\newblock In {\em Proceedings of the 22nd Symposium on Discrete Algorithms
  (SODA)}, pages 560--569, 2011.
\newblock URL: \url{http://dl.acm.org/citation.cfm?id=2133036.2133080}.

\bibitem[RM99]{RazM99}
Ran Raz and Pierre McKenzie.
\newblock Separation of the monotone {NC} hierarchy.
\newblock {\em Combinatorica}, 19(3):403--435, 1999.
\newblock \href {http://dx.doi.org/10.1109/SFCS.1997.646112}
  {\path{doi:10.1109/SFCS.1997.646112}}.

\bibitem[Sim97]{Sim97}
Daniel Simon.
\newblock On the power of quantum computation.
\newblock {\em SIAM Journal on Computing}, 26(5):1474--1483, 1997.
\newblock \href {http://dx.doi.org/10.1137/S0097539796298637}
  {\path{doi:10.1137/S0097539796298637}}.

\bibitem[Yao77]{Yao77}
Andrew Yao.
\newblock Probabilistic computations: Toward a unified measure of complexity.
\newblock In {\em Proceedings of the 18th Symposium on Foundations of Computer
  Science (FOCS)}, pages 222--227, 1977.
\newblock \href {http://dx.doi.org/10.1109/SFCS.1977.24}
  {\path{doi:10.1109/SFCS.1977.24}}.

\end{thebibliography}

\end{document}